\documentclass[journal,12pt,onecolumn,draftclsnofoot,]{IEEEtran}
\usepackage{algorithm,algorithmic}

\usepackage{amsthm,amsmath,amssymb}
\usepackage{eucal}
\usepackage{xifthen}
\usepackage{mathtools}
\usepackage{enumerate}
\usepackage{microtype}
\usepackage{xspace}
\usepackage{bm}
\usepackage{cite}
\usepackage{fancyhdr}
\usepackage{lastpage}
\usepackage[small]{caption}
\usepackage{xcolor}
\usepackage{algorithm}
\usepackage{booktabs}
\usepackage{multirow}
\usepackage{subfigure}
\allowdisplaybreaks

\long\def\comment#1{}

\newfont{\bbb}{msbm10 scaled 700}

\newfont{\bb}{msbm10 scaled 1100}

\newcommand{\mbs}[1]{\bm{#1}}

\newcommand{\mat}[1]{{\uppercase{\mbs{#1}}}}

\renewcommand{\Re}[1][]{\ifthenelse{\isempty{#1}}{\operatorname{Re}}{\operatorname{Re}\left(#1\right)}}
\renewcommand{\Im}[1][]{\ifthenelse{\isempty{#1}}{\operatorname{Im}}{\operatorname{Im}\left(#1\right)}}

\newcommand{\Tm}{\mat{t}}

\newcommand{\CN}[1][]{\ifthenelse{\isempty{#1}}{\mathcal{N}_{\mathbb{C}}}{\mathcal{N}_{\mathbb{C}}\left(#1\right)}}
\renewcommand{\P}[1][]{\ifthenelse{\isempty{#1}}{\mathbb{P}}{\mathbb{P}\left(#1\right)}}
\newcommand{\E}[1][]{\ifthenelse{\isempty{#1}}{\mathbb{E}}{\mathbb{E}\left(#1\right)}}
\renewcommand{\det}[1][]{\ifthenelse{\isempty{#1}}{\mathrm{det}}{\mathrm{det}\left(#1\right)}}
\newcommand{\trace}[1][]{\ifthenelse{\isempty{#1}}{\mathrm{tr}}{\mathrm{tr}\left(#1\right)}}
\newcommand{\rank}[1][]{\ifthenelse{\isempty{#1}}{\mathrm{rank}}{\mathrm{rank}\left(#1\right)}}
\newcommand{\diag}[1][]{\ifthenelse{\isempty{#1}}{\mathrm{diag}}{\mathrm{diag}\left(#1\right)}}
\newcommand{\blkdiag}[1][]{\ifthenelse{\isempty{#1}}{\mathrm{blkdiag}}{\mathrm{blkdiag}\left(#1\right)}}

\renewcommand{\Re}{{\rm Re}}
\renewcommand{\Im}{{\rm Im}}

\allowdisplaybreaks
\newcommand{\st}{{\rm s.t.}}
\DeclareMathAlphabet{\mathcal}{OMS}{cmsy}{m}{n}

\newcommand{\eqdef}{\triangleq}

\newtheorem{remark}{Remark}
\newtheorem{theorem}{Theorem}
\newtheorem{corollary}{Corollary}
\newtheorem{lemma}{Lemma}

\usepackage{hyperref}
\hypersetup{
    bookmarks=true,         
    unicode=false,          
    pdftoolbar=true,        
    pdfmenubar=true,        
    pdffitwindow=false,     
    pdfstartview={FitH},    
    pdfnewwindow=true,      
    colorlinks=true,       
    linkcolor=red,          
    citecolor=green,        
    filecolor=magenta,      
    urlcolor=cyan           
}

\usepackage[utf8]{inputenc} 
\usepackage[T1]{fontenc}
\usepackage{url}
\usepackage{ifthen}
\usepackage{cite}
\usepackage{amsmath} 

\long\def\longdelete#1{}

\longdelete{
    \textheight=200mm \textwidth=140mm
    \oddsidemargin=8mm
    \evensidemargin=8mm
}

\begin{document}
\title{Topological Interference Management with Adversarial Topology Perturbation: \\An Algorithmic Perspective}

\author{
\IEEEauthorblockN{Ya-Chun Liang, Chung-Shou Liao, and Xinping Yi
}

\thanks{This paper has been presented in part at IEEE International Symposium on Information Theory, Melbourne, Australia, 2021~\cite{9518022}.} 
\thanks{
Y. Liang is with Department of Industrial Engineering and Engineering Management, National Tsing Hua University, Hsinchu 30013, Taiwan, ROC, and is with Department of Electrical Engineering \& Electronics, University of Liverpool, Liverpool L69 3GJ, England, UK. Email: ycliang512@gapp.nthu.edu.tw, ya-chun.liang@liverpool.ac.uk.}
\thanks{C. Liao is with Department of Industrial Engineering and Engineering Management, National Tsing Hua University, Hsinchu 30013, Taiwan, ROC. Email: csliao@ie.nthu.edu.tw.}
\thanks{X. Yi is with Department of Electrical Engineering \& Electronics, University of Liverpool, Liverpool L69 3GJ, England, UK. Email: xinping.yi@liverpool.ac.uk.}}

\maketitle

\begin{abstract}
In this paper, we consider the topological interference management (TIM) problem in a dynamic setting, where an adversary perturbs network topology to prevent the exploitation of sophisticated coding opportunities (e.g., interference alignment). Focusing on a special class of network topology -- chordal networks -- we investigate algorithmic aspects of the TIM problem under adversarial topology perturbation. In particular, given the adversarial perturbation with respect to edge insertion/deletion, we propose a dynamic graph coloring algorithm that allows for a \emph{constant} number of re-coloring updates against each inserted/deleted edge to achieve the information-theoretic optimality. This is a sharp reduction of the general graph re-coloring, whose optimal number of updates scales as the size of the network, thanks to the delicate exploitation of the structural properties of {chordal graph classes}. 
\end{abstract}



\section{Introduction}

Due to the broadcast nature of wireless communications, interference between multiple concurrent information flows is a major obstacle limiting multi-user network capacity, which is even more severe in massive Internet of Things (IoT) networking. Interference management (IM) is, therefore, vital for reliable and high-rate communications. While conventional IM techniques in emerging 5G standards (e.g., LTE-M) promise some guarantees, a major obstacle is the difficulty of acquiring timely and accurate channel state information (CSI) at the transmitters, especially in large networks with high-mobility IoT devices. 

To resolve this issue, one of the most promising techniques is topological interference management (TIM) \cite{jafar2013topological}, in which the CSI requirement is considerably reduced. It only requires knowledge of network connectivity patterns, i.e., a graph with edges indicating whether or not communication between two vertices is feasible. 
%
%
To be specific, the TIM problem \cite{jafar2013topological} is interested in the characterization of the degrees of freedom (DoF) region of partially-connected wireless networks with no channel state information at the transmitters (no CSIT) beyond the network topology. 
With no CSIT, it seems the transmitters are restricted to conventional multiple access techniques, e.g., time-division multiple access (TDMA), frequency reuse, and code-division multiple access (CDMA).
Remarkably, it has been evidenced in \cite{jafar2013topological} that, given only the pure knowledge of the network topology (i.e., partial connectivity graph) but not the exact channel realizations, advanced techniques such as interference alignment are able to exploit coding opportunities and achieve substantial improvement over TDMA, frequency reuse, and CDMA. 

Since then, TIM has attracted extensive attention, and the list of follow-up works is growing, including multi-level TIM \cite{geng2013multilevel}, multi-antenna TIM \cite{sun2014topological}, TIM with alternating topology \cite{sun2013topological, gherekhloo2013topological}, TIM with transmitter/receiver cooperation \cite{yi2015topological,yi2018topological}, TDMA TIM \cite{maleki2013optimality,yi2018tdma}, TIM with topology uncertainty \cite{OPTIM}, and TIM with confidential messages \cite{mutangana2020topological}, and many others (e.g., \cite{naderializadeh2014interference,gao2014topological,shi2016low,aquilina2016degrees,yang2017topological,el2017topological,davoodi2018network,doumiati2019framework}). The state-of-the-art TIM focuses on the {\em static} setting, where the network topology is kept unchanged during communication. However, as the applications of robotic and vehicular communications emerge, it is 
unlikely to maintain the network topology unchanged.
%

To deal with the dynamicity introduced by physical environment, we revisit the TIM problem with a dynamic setting.
It is a challenging problem for TIM under a dynamic setting, where the network topology may be time-varying.
Let us take the high-mobility scenarios for example, where the vehicles increasingly take part in IoT networking, so that the network connectivity pattern may change over time. 
The conventional wisdom is to re-design coding schemes once network topology is changed, because different coding techniques may be required to achieve the information-theoretic optimality. However, coding re-design is time-consuming, especially for large-scale networks.
For the sake of tractability, we take a first step to look at an adversarial perturbation model, where there exists an adversary who changes the network topology slightly, by edge insertion or deletion, to prevent the exploitation of coding opportunities.

%
Let us take the network depicted in Figure \ref{fig:convex} as an example.
For this network topology, recognized as a convex network \cite{maleki2013optimality}, it has been shown that TDMA is information-theoretically optimal in terms of sum DoF. In particular, a greedy message scheduling scheme that activates pairs $S_1 \rightarrow D_1, S_3 \rightarrow D_5, S_5 \rightarrow D_7, S_6 \rightarrow D_{10}$ and $S_8 \rightarrow D_{14}$ from left to right, achieves the optimal sum DoF value of 5.
%
\begin{figure}[htb]
 \centering
\includegraphics[width=0.8\columnwidth]{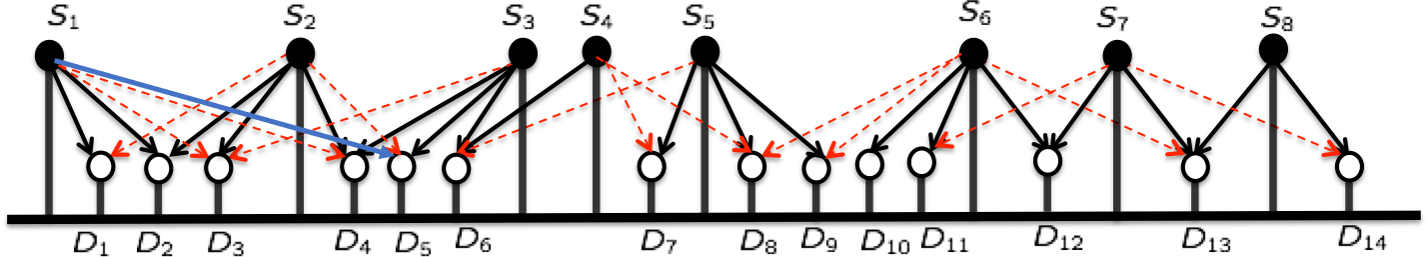}
\caption{A one-dimensional convex network with partial connectivity \cite{maleki2013optimality}. The solid black edges indicate independent desired messages, the dashed red edges are interference, and the solid blue edge is the adversarial topology perturbation with edge insertion as an additional interfering link. 
}
\label{fig:convex}
\end{figure}
%
With a slight change of the network topology, e.g., edge insertion of
$S_1 \rightarrow D_5$ as an additional interfering link, it appears we cannot schedule $S_3 \rightarrow D_5$ after the activation of $S_1 \rightarrow D_1$. Now that the original message scheduling is conflicting due to edge insertion, 
shall we restart the scheduling from $S_1 \rightarrow D_1$, followed by $S_3 \rightarrow D_6$, and so on?


In general, given the adversarial network perturbation, two questions then arise as to (1) whether or not we should re-schedule all messages to adapt such a slight change of topology; and (2) whether or not such re-scheduling still yields the optimal solution. 
%
To deal with these issues, in this paper we work towards the TIM problem in a dynamic setting by starting with a special family of network topologies -- chordal networks \cite{yi2018tdma} -- for which message scheduling (i.e., TDMA) is information-theoretically optimal. The fundamental structural property that makes other sophisticated coding schemes unnecessary is the {chordality}.
{
Precisely, 
a network has the chordal property, or it is called a \emph{chordal network}, if the network topology does not contain chordless cycles with length greater than four \cite{yi2018tdma}. 
Note that the network topology is a bipartite graph with sources and destinations being two sets of vertices such that any cycle involves the same number of sources and destinations. Therefore, there only exist cycles with length four, six, eight, etc, as length two cycles reduce to an edge. Moreover, the chordless cycles do not contain any chord in the cycles, and are deemed as the shortest cycles within those nodes.
}

As a starting point, an adversary attempts to change network topology by edge insertion/deletion with the {chordal property} maintained, so that no advanced coding opportunities can be exploited \cite{yi2018tdma}. Under this setting, we propose a dynamic graph coloring algorithm on a message conflict graph to robustify the message scheduling scheme against such perturbation. 
{
Specifically, a message conflict graph is an undirected graph where a vertex represents a message from a source to a destination in the network topology graph, and an edge exists between two vertices if and only if the two messages conflict with each other when transmitting simultaneously. That is, two conflicting messages are either sent from the same source, received by the same destination, or the transmission of one message interferes that of the other one over the same time/frequency/space resource.
} 
Under such adversarial setting in chordal networks, we prove that the proposed algorithm maintains the information-theoretic optimality under the TIM setting with a constant number of re-coloring updates regardless of the sizes of the network or message set. Leveraging graph theoretic tools, the proof technique deliberately exploits the structural properties of {chordal graph classes}. 

{
The rest of this paper is organized as follows. In the next section, we present the modeling and problem formulation of the dynamic TIM. In Section III, we summarize the main results including the information-theoretic optimality and the updating complexity of the proposed algorithms. The graph models and algorithms are detailed in Section IV, and detailed proofs of the main results can be found in Section V, followed by the Conclusion in Section VI.
}

\section{Model and Problem Formulation}
\subsection{Topological Interference Management (TIM)}
The general TIM problem considers a \emph{partially-connected} wireless interference network with $M$ sources (transmitters), labeled as $S_1$, $S_2$, $\dots$, $S_M$, and $N$ destinations (receivers), labeled as $D_1$, $D_2$, $\dots$, $D_N$, each equipped with one single antenna.
The received signal for $D_j$ at time instant $t$ is given by
\begin{align}
Y_j(t) = \sum_{i=1}^M t_{ji} h_{ji}(t) X_i(t) + Z_j(t)
\end{align}
where $X_i(t)$ is the transmitted signal from $S_i$, and $Z_j(t)$ is the additive Gaussian noise with zero-mean and unit-variance, $h_{ji}(t)$ is the channel coefficient between $S_i$ and $D_j$, and $t_{ji}=1$ if $S_i$ and $D_j$ are connected, and 0 otherwise. The network topology (i.e., partial connectivity patterns) can be represented by a binary matrix $\Tm=[t_{ji}]_{N \times M}$, which is assumed to be known at all sources and destinations.
The TIM problem under the {\em static} setting \cite{jafar2013topological} assumes the actual channel realizations $\{h_{ji}(t), \forall i,j,t\}$ are not available at the sources, and the network topology $\Tm$ is fixed during communication.

As to the message set, we follow the general multiple unicast model in \cite{jafar2013topological,yi2018tdma}, where each message originates from one unique source and intends for one unique destination. All messages are independent. As such, each source may have multiple independent messages that could intend for multiple destinations and each destination may desire multiple independent messages that could originate from multiple sources.

Under the TIM setting, the commonly used figures of merit include symmetric degrees of freedom (DoF), sum DoF, and DoF region, which follow the standard definition in the literature (e.g., \cite{jafar2013topological,yi2018tdma}). Roughly speaking, the DoF value is the maximum number of interference-free sub-streams that can be transmitted at high signal-to-noise ratio (SNR). Among many promising DoF results for TIM,
the DoF characterization in \cite{maleki2013optimality,yi2018tdma} reveals the limitation of coding in chordal networks, which
is of particular relevance to this work.

\subsection{TIM in Chordal Networks}

When treating sources and destinations as vertices, and their connectivity as edges, the partially-connected network is a bipartite graph (also referred to as network topology graph), and the topology matrix $\Tm$ is the adjacency matrix.

{
Chordal networks are a family of network topology recognized as chordal bipartite graphs, which is a class of bipartite graphs such that there either exist no cycles or no chordless cycles with length larger than four. A {\em cycle} is a closed loop formed by a set of vertices and edges. 
The number of vertices in the cycle is referred to as the {\em length} of the cycle.  
A {\em chordless cycle} is a cycle with no edges between any non-consecutive vertices.
}

It has been proved in \cite{yi2018tdma} that TDMA is the information-theoretically optimal coding scheme in chordal networks, and achieves all-unicast DoF region.
TDMA can be interpreted as a general message scheduling scheme such that any interfering messages are not transmitted simultaneously. That is, any two conflicting messages that originate from the same source, intended for the same destination, or one message from a source interferes the other message's destination, will be scheduled at two distinct time slots.
TDMA on network topology graphs can be alternatively done by vertex coloring on the message conflict graph constructed from the network topology graph to capture the conflict between messages.

\subsection{Vertex Coloring for TIM}
\label{sec:graph_def}
To describe vertex coloring, we first introduce some definitions in graph theory.
\subsubsection{Graph Definitions}

Given a graph $G=(V,E)$,
a {clique} is a subgraph of $G$ with every two vertices connected with an edge. 
The {clique number} is the maximum possible size of the cliques in $G$.
The maximum clique problem is then defined to solve the problem for finding a largest possible number of vertices that form a clique in a given graph.

Vertex coloring is to assign colors to vertices of a graph so that the adjacent vertices receive different colors. Let us define a mapping for vertex coloring $f_G: V(G) \mapsto \{1,2,\dots,c\}$ such that $f_G(v_i) \neq f_G(v_j)$ if $(v_i,v_j) \in E(G)$, where $c$ is an integer. The chromatic number is the minimum $c$ for all valid vertex color assignments.
The minimum coloring problem is formulated to find a
valid vertex color assignment such that the required number of colors is minimized.
It is an NP-hard problem on a general graph.

A graph is perfect if and only if for every induced subgraph the clique number is equal to the chromatic number \cite{golumbic2004algorithmic}, where the induced subgraph refers to a subgraph of an original graph with a subset of vertices such that all edges associated to the subset of vertices are kept.
{
A graph is {\em chordal} (resp. {\em weakly chordal}) if there is no induced subgraph with chordless cycle of length greater than three (resp. four), i.e., every cycle with length greater than three (resp. four) has a chord.
Both chordal and weakly chordal graphs are a subclass of perfect graphs.}
%
\begin{figure}[htb]
\centering
\includegraphics[scale=0.6]{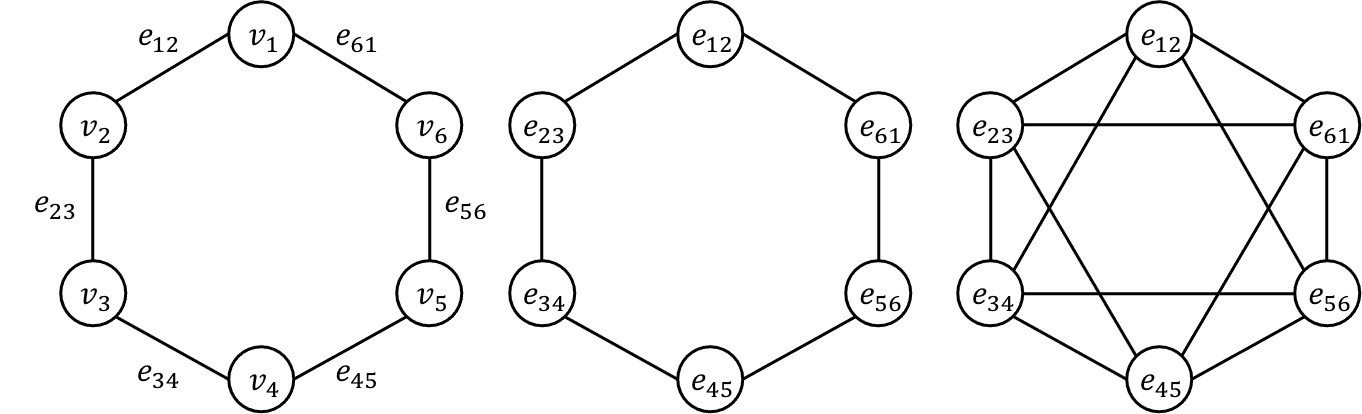}
\caption{(Left) Graph $T$; (Middle) Line graph $T_e$; (Right) Square of line graph $T_e^2$.}
\label{fig:graph_line}
\end{figure}
 
Given a graph $T$,
the {\em line graph} of $T$ is another graph, denoted by $T_e$, such that each vertex of $T_e$ represents an edge of $T$, and any two vertices in $T_e$ are adjacent if and only if their corresponding edges in $T$ have an endpoint in common (as shown in the middle of Fig.~\ref{fig:graph_line}).
The {\em square} of a graph $T_e$, denoted by $T_e^2$, is another graph with the same vertex set as $T_e$, but any two vertices that are adjacent to a common vertex in $T_e$ are also incident to an edge in $T_e^2$ (as shown in the right of Fig.~\ref{fig:graph_line}).
%

\subsubsection{Vertex Coloring on Conflict Graphs}
For the TIM setting with a specified message set, the {\em message conflict graph} is defined to be a graph with vertices 
which represent messages and edges between two vertices if the two messages have a conflict with each other. 
Two messages have a conflict if 
(1) they originate from the same source, 
(2) they intend for the same destination, 
or (3) one source interferes the other destination.

{
It is shown in \cite[Lemma 1]{yi2018tdma} that, for any network topology graph $T$, the square of its line graph, $T_e^2$, is its message conflict graph.
Note that if $T$ is weakly chordal, then $T_e^2$ is also weakly chordal \cite{cameron2003finding}.
For chordal networks, its network topology $T$ which is a chordal bipartite graph is a weakly chordal graph, so that the message conflict graph $T_e^2$ is also weakly chordal.
More specifically, if the network topology $T$ is a biconvex bipartite graph (corresponding to a convex network \cite{maleki2013optimality}), then the message conflict graph $T_e^2$ is a chordal graph. 
}

\begin{figure}[htb]
\centering
\includegraphics[scale=0.6]{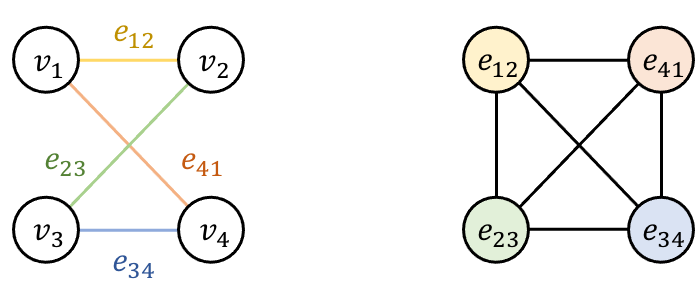}
\caption{(Left) Network topology graph $T$; (Right) Message conflict graph $T_e^2$.}
\label{fig:graph_trans}
\end{figure}

Given such graph properties, \cite{yi2018tdma} shows that vertex coloring on $T_e^2$, equivalently TDMA on $T$, is the information-theoretically optimal coding scheme for the all-unicast TIM problem in chordal networks. 
{
As shown in Fig.~\ref{fig:graph_trans}, the clique number of the message conflict graph, as well as the chromatic number due to perfect graph,
represents the maximum number of messages that cannot be transmitted simultaneously in the network topology graph. That is, the messages should be scheduled at four distinct time slots for interference management.
}

In what follows, our focus will be placed on vertex coloring on the message conflict graph $T_e^2$. For notational simplicity, let us denote hereafter by $G \eqdef T_e^2$ as the message conflict graph. 

%
%
%

\subsection{Adversarial Perturbation Model}

In this work, we consider the dynamic TIM problem in chordal networks with an adversarial perturbation model. The network topology is perturbed by an adversary with edge insertion or deletion, which corresponds to inserted or deleted interfering links. We assume the message demanding patterns are kept unchanged, so that the adversarial topology perturbation leads only to edge insertion/deletion on the message conflict graph $G$ due to added/removed interference, but not to node insertion/deletion in $G$. 
We assume the {chordal property} is still maintained in $G$ after edge insertion/deletion, inspired by the adversary models widely used in the distributed computation literature 
(e.g., \cite{kuhn2010distributed}, \cite{eppstein1999dynamic}) and the goal of the adversary to prevent potential coding gains beyond message scheduling (e.g., TDMA).


{\bf Edge Insertion}: Given the original {message conflict graph $G$}, an inserted edge $(u,v)$ yields a new graph $H=G+(u,v)$,
{where the chordality is retained.}

{\bf Edge Deletion}: Given the original {message conflict graph $H$}, a deleted edge $(u,v)$ yields a new graph $G=H-(u,v)$,
{where the chordality is retained.}

Specifically, the retained chordality is in the sense that if $G$ is chordal (resp. weakly chordal), then $H$ is also chordal (resp. weakly chordal), and vice versa.
Note that $(u,v) \notin E(G)$ and $(u,v) \in E(H)$. 
The goal is to robustify the coding schemes with minimal updating cost.
From an algorithmic perspective, we aim to design robust coloring functions $f_H$ for edge insertion and $f_G$ for edge deletion with minimal re-coloring updates, i.e.,
\begin{align*}
    \min \quad &  \mathrm{diff} (f_G, f_H)\\
    \st \quad &  H=G+(u,v),\\
    & f_G, f_H \text{ are information-theoretically optimal,}
\end{align*}
where $\mathrm{diff} (f_G, f_H)$ is the difference of color assignments of $G$ and $H$ with edge insertion/deletion.







\section{Main Results}

{
The main results are a set of dynamic graph re-coloring algorithms on the message conflict graph against adversarial perturbation of edge insertion/deletion, which can be straightforwardly translated into message scheduling schemes on the network topology graph to robustify TIM coding schemes in chordal networks for both weakly chordal (i.e., Algorithms \ref{alg:weakly_insertion} and \ref{alg:weakly_deletion}) and chordal graphs (i.e., Algorithms \ref{alg:chordal_insertion} and \ref{alg:chordal_deletion}).
}
In what follows, we present its information-theoretic optimality and updating complexity for edge insertion (i.e., Algorithms \ref{alg:weakly_insertion} and \ref{alg:chordal_insertion}), and delegate the proofs to Section \ref{sec:proof}.
{
For edge deletion (i.e., Algorithms \ref{alg:weakly_deletion} and \ref{alg:chordal_deletion}), the similar results apply.
}
{
\begin{theorem}[Optimality]
\label{thm:optimality}
For the TIM setting with adversarial perturbation, Algorithms~\ref{alg:weakly_insertion} and~\ref{alg:chordal_insertion} achieve the information-theoretic optimality in chordal networks with conflict graphs being weakly chordal and chordal, respectively, when the chordal property is maintained with edge insertion.
\end{theorem}
}

The information-theoretic optimality is based on the facts that (1) the chordal property is maintained for both its network topology after edge insertion and the corresponding message conflict graph, such that vertex coloring on conflict graph (i.e., TDMA on network topology) achieves the optimal DoF \cite{yi2018tdma}; and (2) Algorithms~\ref{alg:weakly_insertion} and~\ref{alg:chordal_insertion} achieve optimal vertex coloring on conflict graph in the graph classes of weakly chordal and chordal graphs, respectively, when the chordal property is maintained with edge insertion.
The optimality for edge deletion simply applies, because these two facts still hold where Algorithms \ref{alg:weakly_deletion} and \ref{alg:chordal_deletion} achieve optimal vertex coloring after edge deletion when chordal property is maintained.

From the algorithmic perspective, the key technique that maintains the optimal coloring in Algorithm~\ref{alg:weakly_insertion} for weakly chordal graphs is due to the property 
of \emph{two-pairs} introduced in \cite{hayward1989optimizing}.
Thanks to the property of two-pairs, 
the algorithm proposed in~\cite{hayward1989optimizing} can 
solve a variety of optimization problems in any static weakly chordal graph according to an arbitrary order of a sequence of two-pairs. 
Going beyond the static setting
and taking advantage of the arbitrary ordering, Algorithm~\ref{alg:weakly_insertion}
always achieves the optimal coloring as long as the weakly chordal property is maintained with edge insertion.
%
{
When it comes to chordal graphs,
the essential property 
that maintains the optimal coloring
in Algorithm~\ref{alg:chordal_insertion} is the perfect elimination order (PEO) introduced in~\cite{rose1970triangulated}.
}
To achieve the minimum vertex coloring for chordal graphs (in the static setting), 
one can greedily color all the vertices in the reverse order of the PEO~\cite{gavril1972algorithms}. 
{
Algorithm~\ref{alg:chordal_insertion} develops a dynamic operation adjusting both the perfect elimination order and the coloring locally for edge insertion, and thus achieves the minimum vertex coloring in chordal graphs.
}


In addition to the optimality,
the updating cost of dynamic operations can be upper bounded
{
by a constant that does not scale as the sizes of the network and the message sets.
}

{
\begin{theorem} [Updating complexity]
\label{thm:complexity}
Algorithms~\ref{alg:weakly_insertion} and~\ref{alg:chordal_insertion} have at most a constant number $c$ of re-coloring updates ($c\leq8$) against any single edge insertion to the message conflict graph.
\end{theorem}
}

While in general it requires $\left| V(G) \right|-\left| {K_G} \right|$ iterations to color a weakly chordal graph (see Algorithm~\ref{alg:weakly_static}), 
Algorithm \ref{alg:weakly_insertion} avoids re-coloring all vertices given a small topology perturbation. It turns out that the re-coloring of a constant number of vertices could yield the optimal vertex coloring solution.
For instance, the only replacement of $S_3 \rightarrow D_5$ by $S_3 \rightarrow D_6$ in Fig. \ref{fig:convex} gives us the optimal solution in the perturbed topology. 
The proof is due to the fact that
the message conflict graph is the square of line graph of chordal bipartite network topology, 
and thus
possesses the $(4, 4, 2)$ representation of an intersection (tree) model \cite{golumbic2009intersection}.
As such, the number of two-pairs that require re-coloring is upper-bounded by a constant.
{
This can be applied straightforwardly to chordal graphs as well 
because they are a subclass of weakly chordal graphs.
}

Nevertheless, the time complexity of dynamic operations in weakly chordal graphs is still challenging to analyze. Thus, we take a step back and consider a subclass of weakly chordal graphs -- chordal graphs -- whose time complexity analysis is more feasible thanks to their property of the PEO. 
The PEO of a chordal graph is an ordering of the vertices such that, for each vertex, itself and its neighbors that occur after it in the ordering form a clique.
{
For chordal graphs, coloring vertices greedily
in the reverse order of the PEO yields the optimal coloring, 
which takes $O(n)$ time.
}
As such, given edge insertion, if we can maintain the PEO by adjusting the orders of a small number of vertices, then the updating time complexity can be upper bounded, as the following corollary says.

{
\begin{corollary}
The updating time complexity of Algorithm~\ref{alg:chordal_insertion} is bounded by O($\Delta$), where $\Delta$ denotes the maximum degree among the vertices.
\end{corollary}
}

The updating time complexity is counted for the re-coloring operations due to edge insertion.
Given an existing PEO for a chordal graph, an edge insertion only makes at most $\Delta$ vertices unsatisfied in the PEO. The re-ordering of these vertices can thus maintain a valid PEO. 
Precisely, suppose either one of the endpoints of the newly-inserted edge is no longer a simplicial vertex, i.e., which does not follow the property of the perfect elimination order. In that case, we identify an appropriate vertex from the neighbors of this endpoint and locally modify the order.
{
In other words, searching for a vertex that conforms to the simplicial property for each edge insertion can be guaranteed within the neighborhood, and thus the updating time complexity for the algorithm is bounded by the maximum degree among all vertices. 
}

\section{Graph Model and Algorithm}

\subsection{Notation}
In addition to graph definitions in Section \ref{sec:graph_def}, we formally introduce some notations dedicated to the algorithm.
\begin{itemize}
    \item $V(G)= \{v_1,v_2,...,v_n\}$: The set of vertices of graph $G$. 
    \item {$E(G)= \{e_1,e_2,...,e_m\}$: The set of edges of graph $G$.}
    \item {$N(v)= \{u \in V(G) | (u,v) \in E(G)\}$: The open neighborhood of $v$.} {We use $d(v)=\left| N(v) \right|$ to denote the degree of $v$.}
    \item $N[v]= \{N(v) \cup {v}\}$: The closed neighborhood of $v$. 
    \item $K_G$: The maximum clique over all cliques of graph $G$. Thus, $\left| {K_G} \right|$ is the clique number. 
    \item $f_G$: A vertex coloring function $f_G: V(G) \mapsto \{1,\dots,c\}$
    that assigns integer from $1$ to $c$ to each vertex in $G$, such that $f_G(v_i) \ne f_G(v_j)$ if $(v_i,v_j) \in E(G)$.
    Hence, we use $\left| {f_G} \right|$ to denote the number of colors used in graph $G$ with the coloring function $f_G$.
    \item $\{x,y\}$: A two-pair is a pair of non-adjacent vertices $x$ and $y$ such that every chordless path between them has exactly two edges.
    \item $G(x,y \to z)$: The graph obtained by replacing vertices $x$ and $y$ of $G$ with a vertex $z$, such that $z$ is adjacent to exactly those vertices of $G-\{x,y\}$ that are adjacent to at least one of $\{x,y\}$. Here, the operation $(x,y \to z)$ is denoted by \emph{contraction}.
    \item $S_G$: The two-pair solution order obtained from graph $G$.
    \item {$v_s$: A simplicial vertex is a vertex if $N[v_s]$ is a clique.}
    \item {PEO: The \emph{perfect elimination order} is an ordering $v_x,..., v_y$ such that each $v_i$ is simplicial in the subgraph induced by the vertices $v_x,..., v_i$.}
    \item {$V(G)_{\geq v}$: The set of vertices of graph $G$ whose indices are greater than or equal to $v$.} 
\end{itemize}

In what follows, we describe the proposed dynamic algorithms for weakly chordal graphs with edge insertion and deletion in Sections~\ref{sec:weakly_insertion} and~\ref{sec:weakly_deletion}, and for chordal graphs in Sections~\ref{sec:chordal_insertion} and~\ref{sec:chordal_deletion}. To make the description more accessible, we also illustrate the algorithms with concrete examples, whereas the correctness of the algorithms for general cases is given Section~\ref{sec:proof}.

\begin{algorithm}[htb]
 \caption{Dynamic Update for Edge Insertion {in Weakly Chordal Graphs}}
 \label{alg:weakly_insertion}
 \begin{algorithmic}[1]
 \renewcommand{\algorithmicrequire}{\textbf{Input:}}
 \renewcommand{\algorithmicensure}{\textbf{Output:}}
 \REQUIRE a weakly chordal graph $G$, a two-pair solution order $S_G$ and an edge $(u,v) \notin E(G)$
 \ENSURE a maximum clique $K_H$ and a minimum number of colors $\left| {f_H} \right|$ 
 \STATE insert the edge $(u,v)$ into graph $G$;
 \STATE $H \gets G.insert(u,v)$
 \IF {$(u,v) \notin S_G$}
    \IF {$f_G(u) \neq f_G(v)$}
    \STATE $\left| {f_H} \right| \gets \left| {f_G} \right|$; \hspace{4cm}\%{{Case 1}}
    \ELSIF{$\left| {K_H} \right| = \left| {K_G} \right|$}
    \STATE $\left| {f_H} \right| \gets \left| {f_G} \right|$; \hspace{4cm}\%{{Case 2-1}}
    \ELSE
    \STATE $\left| {f_H} \right| \gets \left| {f_G} \right|+1$; \hspace{3.5cm}\%{{Case 2-2}}
    \ENDIF
 \ELSIF{$\left| {K_H} \right| = \left| {K_G} \right|$}
 \STATE $\left| {f_H} \right| \gets \left| {f_G} \right|$; \hspace{4.35cm}\%{{Case 3-1}}
 \ELSE
 \STATE $\left| {f_H} \right| \gets \left| {f_G} \right|+1$;
 \hspace{3.8cm}\%{{Case 3-2}}
 \ENDIF
 \end{algorithmic}
\end{algorithm}




\subsection{Dynamic Algorithm for Edge Insertion {in Weakly Chordal Graphs}}
\label{sec:weakly_insertion}

\begin{algorithm}[htb]
\caption{OPT(G): Static Procedure for Minimum Coloring {in Weakly Chordal Graphs}}
 \label{alg:weakly_static}
 \begin{algorithmic}[1]
 \renewcommand{\algorithmicrequire}{\textbf{Input:}}
 \renewcommand{\algorithmicensure}{\textbf{Output:}}
 \REQUIRE a weakly chordal graph $G$
 \ENSURE a maximum clique $K_G$, a minimum number of colors $\left| {f_G} \right|$ and a two-pair solution order $S_G$
 \IF {$G$ has no two-pair $\{x,y\}$}
 \STATE $K_G \gets V(G)$;
    \FOR {$i=1$ to $n$}
    \STATE $f_G(v_i) \gets i$;
    \ENDFOR
 \STATE \textbf{STOP}
 \ELSE
 \STATE $S_G \gets S_G+\{x,y\}$
 \STATE $J \gets G(x,y \to z)$;
 \STATE $K_J,f_J,S_J \gets OPT(J)$;
 \IF {$z \notin K_J$}
 \STATE $K_G \gets K_J$;
 \ELSIF{$x$ is adjacent to all of $K_J-\{z\}$}
 \STATE $K_G \gets K_J-\{z\}+\{x\}$;
 \ELSE
 \STATE $K_G \gets K_J-\{z\}+\{y\}$;
 \ENDIF
 \STATE $f_G(x) \gets f_G(y) \gets f_J(z)$;
 \FOR {each $v_i \in J-\{x,y\}$}
 \STATE $f_G(v_i) \gets f_J(v_i)$;
 \ENDFOR
 \ENDIF
 \end{algorithmic} 
\end{algorithm}

In the static setting,
the algorithm proposed in~\cite{hayward1989optimizing} (reproduced in Algorithm~\ref{alg:weakly_static} for self-containedness) can solve a number of optimization problems for weakly chordal graphs, relying on the property that every weakly chordal graph is either a clique or has a two-pair.
Given a weakly chordal graph other than a clique, the algorithm repeatedly finds a two-pair and reduces the graph size by one by contracting the two-pair into one vertex, until there is no two-pair. 
Note that based on the assumption, the weakly chordal property is maintained when contracting each two-pair. 
Here the derived sequence of two-pairs is recorded in $S_G$.
The original graph is finally transformed into a clique by repeated contractions, 
and thus the coloring of the clique can be trivially obtained.
Meanwhile,
the solution of the maximum clique problem is transformed to the minimum coloring problem to obtain
\emph{chromatic number}.

We further consider the dynamic operation in which edge insertion is applied through
adversarial perturbation.
Without loss of generality, assume we obtain a two-pair solution order $S_G$ from Algorithm~\ref{alg:weakly_static} for a weakly chordal graph $G$.
After inserting an edge $(u,v)$ into the graph, we deal with the changes that would occur according to the two-pair solution order.
Through our proposed dynamic algorithm, 
i.e. Algorithm~\ref{alg:weakly_insertion}, 
we aim at quickly finding the two-pairs
with the known sequence $S_G$ for updating, including the ones that cannot be contracted 
and 
the others needed to be contracted due to the inserted edge, instead of greedily re-searching for new two-pairs by Algorithm~\ref{alg:weakly_static}.
Note that the two-pair solution order after the insertion of $(u,v)$, i.e., $S_H$, will be specified in the details of each case.

We mainly divide the scenario into three cases according to two conditions.
One is whether the inserted edge exists in the original two-pair solution order $S_G$ of graph $G$, and the other is whether the two endpoints of the inserted edge (i.e. vertex $u$ and $v$) have the same color in $f_G$.

{\bf Case 1}: $(u,v) \notin S_G$ and $f_G(u) \neq f_G(v)$, where $(u,v)=(v_3,v_4)$ (as shown in Fig.~\ref{fig:c1}).
In the case, since $f_G(u)$ is not equal to $f_G(v)$, the maximum clique $K_H$ remains the same as $K_G$. Hence, $\left| {f_H} \right|$ is equal to $\left| {f_G} \right|$, even though $S_H$ may not be the same as $S_G$.

\begin{figure}[htb]
\centering
\includegraphics[scale=0.5]{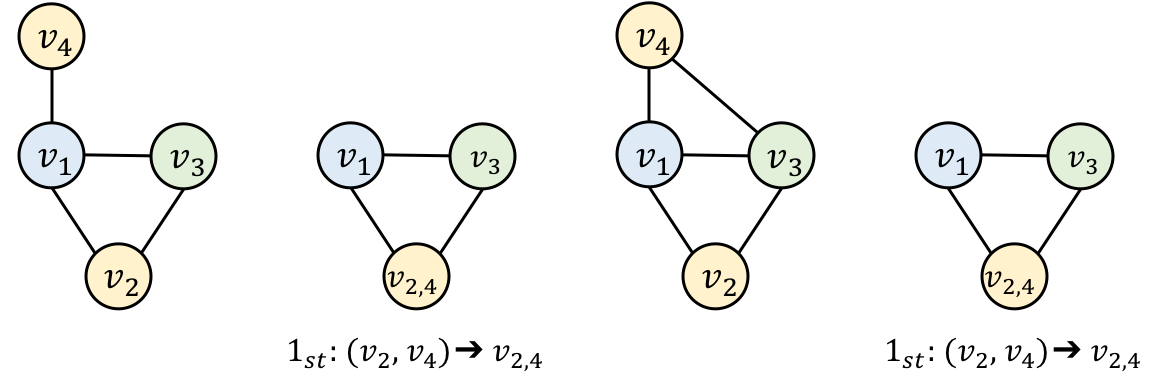}
\caption{The original graph $G$, graph $H$ and the contractions.}
\label{fig:c1}
\end{figure}

{\bf Case 2}: $(u,v) \notin S_G$ and $f_G(u) = f_G(v)$.
If $f_G(u)$ and $f_G(v)$ are the same in graph $G$, either of them is forced to update to a new value due to the inserted edge.
We consider two subcases depending on whether $\left| {f_H} \right|$ is equal to $\left| {f_G} \right|$, concerning the clique size of graph $H$.

\underline{Case 2-1}: $\left| {K_H} \right| = \left| {K_G} \right|$ (as shown in Fig.~\ref{fig:c2-1-G} and Fig.~\ref{fig:c2-1-H}).
After inserting the edge $(u,v)$, if $\left| {K_H} \right|$ remains the same, then $\left| {f_H} \right|$ is equal to $\left| {f_G} \right|$.
Note that $S_H$ may not be the same as $S_G$; that is, the indices of the vertices may be changed.
    
\begin{figure}[htb]
\centering
\includegraphics[scale=0.5]{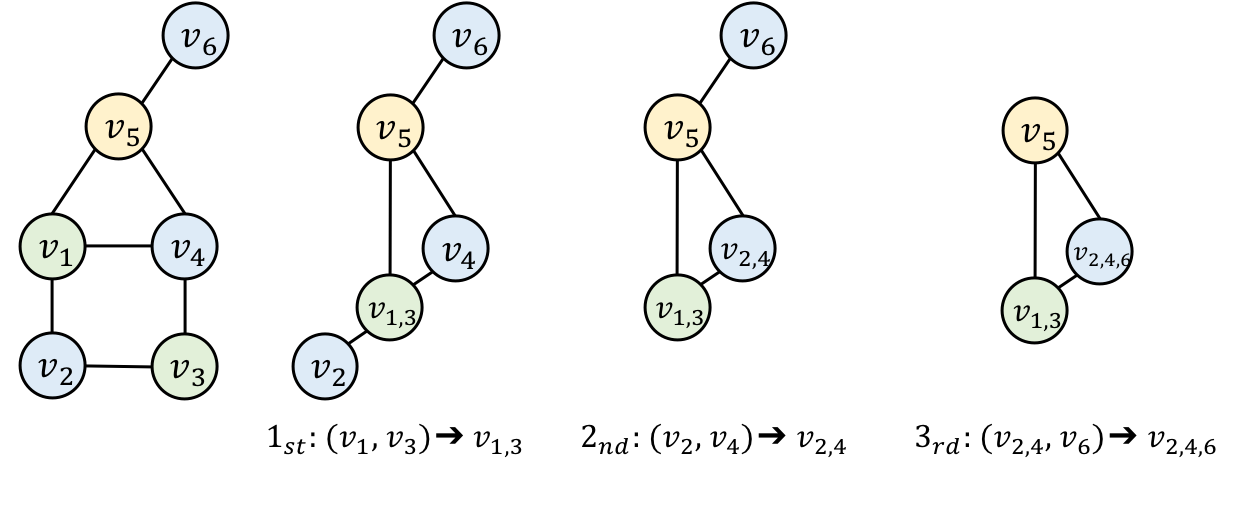}
\caption{The original graph $G$ and its contractions (Case 2-1).}
\label{fig:c2-1-G}
\end{figure}

\begin{figure}[htb]
\centering
\includegraphics[scale=0.5]{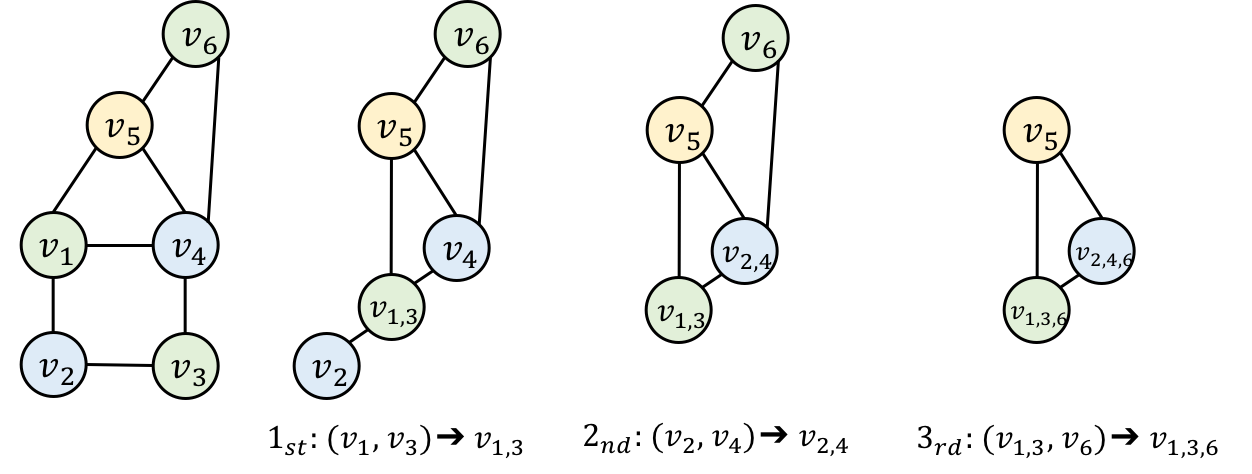}
\caption{The graph $H$ and its contractions (Case 2-1).}
\label{fig:c2-1-H}
\end{figure}
    
\underline{Case 2-2}: $\left| {K_H} \right| \neq \left| {K_G} \right|$ (as shown in Fig.~\ref{fig:c2-2-G} and Fig.~\ref{fig:c2-2-H}).
In this case, if $\left| {K_H} \right|$ and $\left| {K_G} \right|$ are not equal, we have $\left| {K_H} \right| = \left| {K_G} \right|+1$. 
In other words, the contractions of graph $H$ is reduced by one due to the inserted edge $(u,v)$, saying $\left| {S_H} \right| = \left| {S_G} \right|-1$.
Hence, $\left| {f_H} \right| = \left| {f_G} \right|+1$.
    
\begin{figure}[htb]
\centering
\includegraphics[scale=0.5]{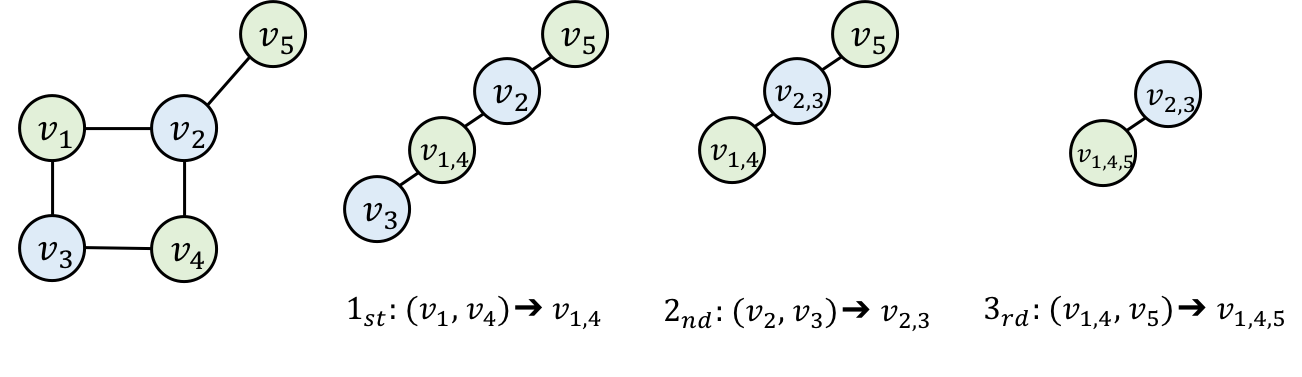}
\caption{The original graph $G$ and its contractions (Case 2-2).}
\label{fig:c2-2-G}
\end{figure}

\begin{figure}[htb]
\centering
\includegraphics[scale=0.5]{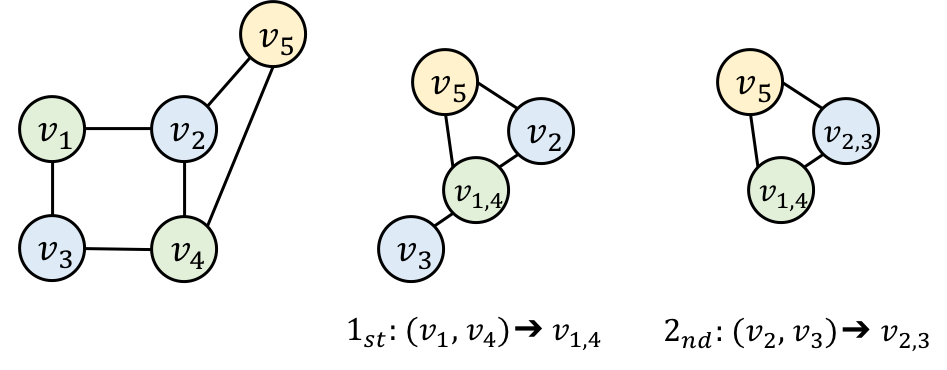}
\caption{The graph $H$ and its contractions (Case 2-2).}
\label{fig:c2-2-H}
\end{figure}


{\bf Case 3}: $(u,v) \in S_G$.
If the inserted edge $(u,v)$ belongs to $S_G$, it implies $f_G(u) = f_G(v)$.
The color of either of the two vertices has to be changed.
Since $u$ and $v$ are adjacent now, 
not only vertex $u$ and vertex $v$ cannot be contracted, but also the related vertices concerning the inserted edges.
Thus, we divide it into two subcases to determine whether $\left| {f_H} \right|$ is equal to $\left| {f_G} \right|$, concerning the clique size of graph $H$.

\underline{Case 3-1}: $\left| {K_H} \right| = \left| {K_G} \right|$ (as shown in Fig.~\ref{fig:c3-1-G} and Fig.~\ref{fig:c3-1-H}).
Although $\left| {f_H} \right|$ is equal to $\left| {f_G} \right|$ eventually in this case, $S_H$ is different from $S_G$.
As mentioned, the contraction is an action performing $(x,y \to z)$.
Therefore, in this case, the two-pair $\{u,v\}$ cannot be contracted to a single vertex.
Concerning the two-pairs in $S_G$, 
the affected ones, i.e. the vertices that are incident to the edge $(u,v)$ and related to the contracted vertex $v_{u,v}$, can no longer be a two-pair.
    
\begin{figure}[htb]
\centering
\includegraphics[scale=0.5]{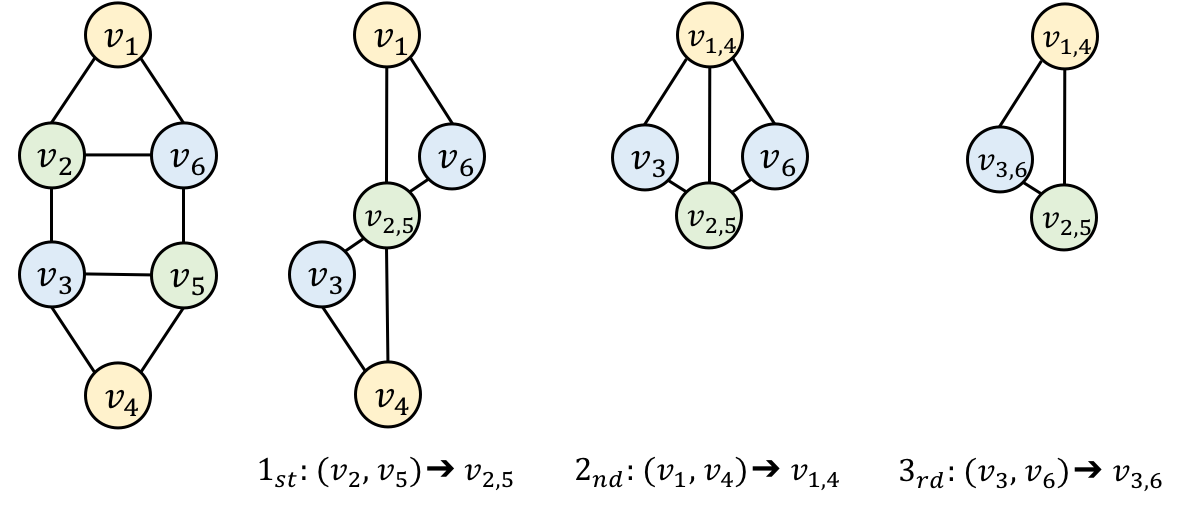}
\caption{The original graph $G$ and its contractions (Case 3-1).}
\label{fig:c3-1-G}
\end{figure}

\begin{figure}[htb]
\centering
\includegraphics[scale=0.5]{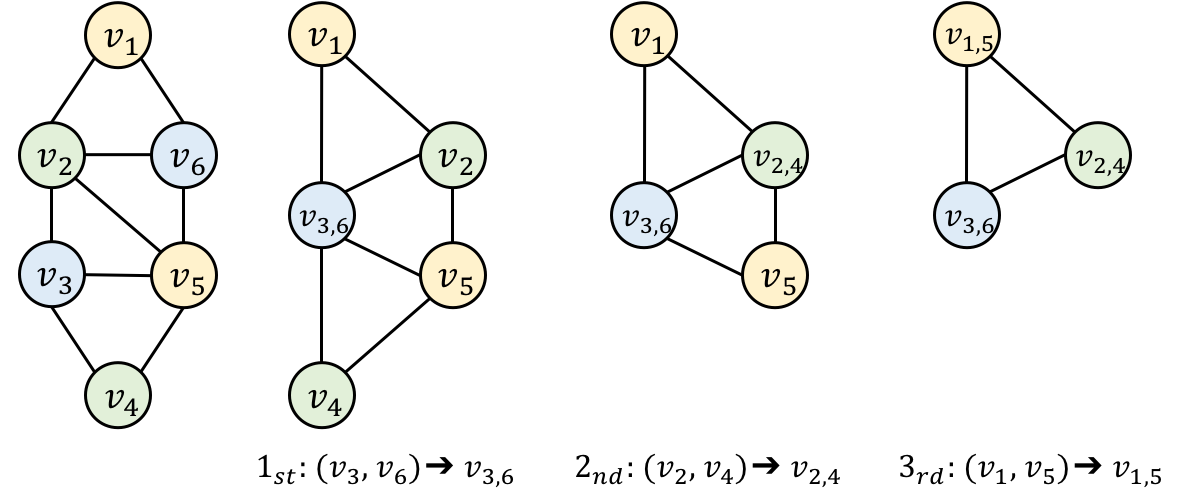}
\caption{The graph $H$ and its contractions (Case 3-1).}
\label{fig:c3-1-H}
\end{figure}

Algorithm~\ref{alg:weakly_insertion} deals with such two-pairs as follows: 
First, we obtain the affected two-pairs by referring to $S_G$. Next we follow the rule of Algorithm~\ref{alg:weakly_static} to re-search for the new two-pairs that are related to the vertices of these affected pairs. 
    
Here we give an illustration to show the scenario. 
As shown in Fig.~\ref{fig:c3-1-G}, $S_G$ contains $(v_2,v_5)$, $(v_1,v_4)$ and $(v_3,v_6)$.
In addition, $V(K_G) = \{v_{1,4},v_{2,5},v_{3,6}\}$ and $\left| {f_G} \right|=3$.
Assume the edge $(v_2,v_5)$ is inserted to $G$ (as shown in Fig. \ref{fig:c3-1-H}).
The affected two-pairs are $(v_2,v_5)$ and $(v_1,v_4)$ due to the inserted edge $(v_2,v_5)$. The reason why $v_1$ and $v_4$ is no longer a two-pair is because the chordless path between $v_1$ and $v_4$ has more than two edges.
Graph $H$ is then contracted through the two-pair $(v_3,v_6)$ in $S_G$.
Next, we use Algorithm~\ref{alg:weakly_static} to search for the new two-pairs adjacent to the vertices of the affected two-pairs, i.e. $v_2$ and $v_1$. 
Thus, $(v_2,v_4)$ and $(v_1,v_5)$ are then contracted to $v_{2,4}$ and $v_{1,5}$, respectively.
Finally, $V(K_H) = \{v_{3,6},v_{2,4},v_{1,5}\}$ and $\left| {K_H} \right| = \left| {K_G} \right|$. 
Therefore, $\left| {f_H} \right|$ remains the same as $\left| {f_G} \right|$.
    
\underline{Case 3-2}: $\left| {K_H} \right| \neq \left| {K_G} \right|$ (as shown in Fig.~\ref{fig:c3-2-G} and Fig.~\ref{fig:c3-2-H}).
In this case, $\left| {K_H} \right| = \left| {K_G} \right|+1$ due to the edge insertion. 
To put it in another way, the number of contractions of graph $H$ is one time less than that of graph $G$.
Since $\left| {S_H} \right| = \left| {S_G} \right|-1$, $\left| {f_H} \right| = \left| {f_G} \right|+1$.
The two-pair $\{u,v\}$ cannot be contracted to a single vertex in this case.
The affected two-pairs are also obtained by referring to $S_G$, including the edge $(u,v)$, the vertices that are incident to the edge $(u,v)$ and related to the contracted vertex $v_{u,v}$.
After that, the new two-pairs that are related to the vertices of these affected pairs are found and contracted until $K_H$ is derived.

\begin{figure}[htb]
\centering
\includegraphics[scale=0.5]{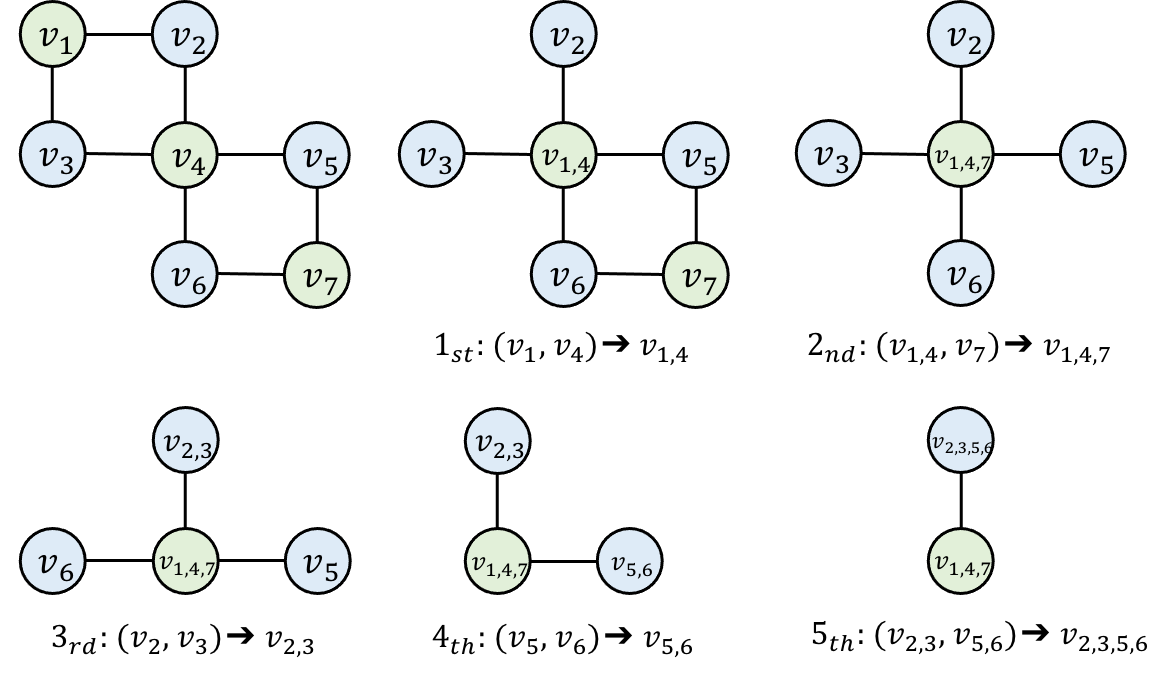}
\caption{The original graph $G$ and its contractions (Case 3-2).}
\label{fig:c3-2-G}
\end{figure}

\begin{figure}[htb]
\centering
\includegraphics[scale=0.5]{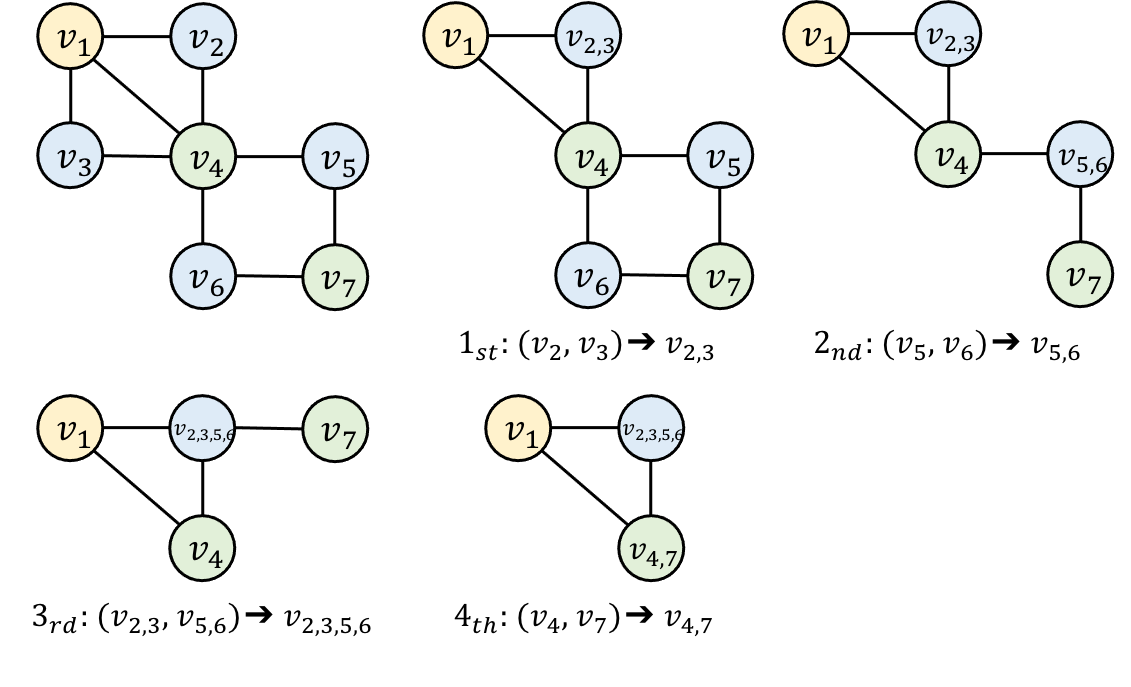}
\caption{The graph $H$ and its contractions (Case 3-2).}
\label{fig:c3-2-H}
\end{figure}

\subsection{Dynamic Update for Edge Deletion {in Weakly Chordal Graphs}}
\label{sec:weakly_deletion}
In Algorithm \ref{alg:weakly_insertion}, we discussed the cases according to two mechanisms: one is whether the inserted edge exists in the two-pair solution order obtained from the original graph, and the other is whether the two endpoints of the inserted edge have the same color in the original graph. 
However, the two mechanisms are unnecessary to consider, concerning edge deletion. The reason is that if $(u, v) \in E(H)$, $f_H(u)$ is not equal to $f_H(v)$. In addition, $u$ and $v$ will never form a two-pair, thus not being contained in $S_H$.
Therefore, when deleting an edge $(u, v)$ from graph $H$, we only consider whether the clique size is reduced by one or remains the same.

\begin{algorithm}[htb]
 \caption{Dynamic Update for Edge Deletion {in Weakly Chordal Graphs}}
 \label{alg:weakly_deletion}
 \begin{algorithmic}[1]
 \renewcommand{\algorithmicrequire}{\textbf{Input:}}
 \renewcommand{\algorithmicensure}{\textbf{Output:}}
 \REQUIRE a weakly chordal graph $H$
 \ENSURE a maximum clique $K_G$ and a minimum number of colors $\left| {f_G} \right|$
 \STATE delete the edge $(u,v)$ from graph $H$;
 \STATE $G \gets H.delete(u,v)$
 \IF {$\left| {K_G} \right| = \left| {K_H} \right|$}
 \STATE $\left| {f_G} \right| \gets \left| {f_H} \right|$;
 \hspace{4cm}\%{{Case 1}}
 \ELSE
 \STATE $\left| {f_G} \right| \gets \left| {f_H} \right|-1$;
 \hspace{3.4cm}\%{{Case 2}}
 \ENDIF
 \end{algorithmic}
\end{algorithm}

Without loss of generality, suppose we obtain a two-pair solution order $S_H$ from Algorithm \ref{alg:weakly_static} for a weakly chordal graph $H$.
Here is the idea about how we deal with the dynamic update of the edge deletion of $(u,v)$ from graph $H$ according to Algorithm \ref{alg:weakly_deletion}.  
We first contract the two-pairs based on the known sequence $S_H$, and then re-search for a new two-pair, if needed, due to the affected edge $(u,v)$ in $H$ by Algorithm~\ref{alg:weakly_static}. 

{\bf Case 1}: $\left| {K_G} \right| = \left| {K_H} \right|$.
Suppose we have Fig. \ref{fig:c3-1-H} as an input weakly chordal graph $H$. Following $S_H$ obtained from Algorithm \ref{alg:weakly_static}, we then contract the two-pairs $(v_3,v_6)$, $(v_2,v_4)$ and $(v_1,v_5)$ and obtain $V(K_G) = \{v_{1,5},v_{2,4},v_{3,6}\}$ and $\left| {f_G} \right|=3$, as shown in Fig. \ref{fig:d3-1-G}. 
This is the simple case that $K_G$ remains the same as $K_H$.

{\bf Case 2}: $\left| {K_G} \right| \neq \left| {K_H} \right|$.
We illustrate the other case that $\left| {K_G} \right| \neq \left| {K_H} \right|$ by using Fig.~\ref{fig:c3-2-H} as an input weakly chordal graph $H$. 
We follow $S_H$ and contract the two-pairs $(v_2,v_3)$, $(v_5,v_6)$, $(v_{2,3},v_{5,6})$ and $(v_4,v_7)$. However, since $(v_1,v_4)$ is deleted from $H$, there exists one more two-pair in graph $G$, i.e., $(v_1,v_{4,7})$. The clique $V(K_G)$ is then obtained and $\left| {K_G} \right|$ is reduced by one, as shown in Fig.~\ref{fig:d3-2-G}. Thus, $\left| {f_G} \right|=2$. 
Namely, the number of contractions of graph $G$ is one more than that of graph $H$.
That is, $\left| {S_G} \right| = \left| {S_H} \right|+1$, $\left| {f_G} \right| = \left| {f_H} \right|-1$. 

\begin{figure}[htb]
\centering
\includegraphics[scale=0.5]{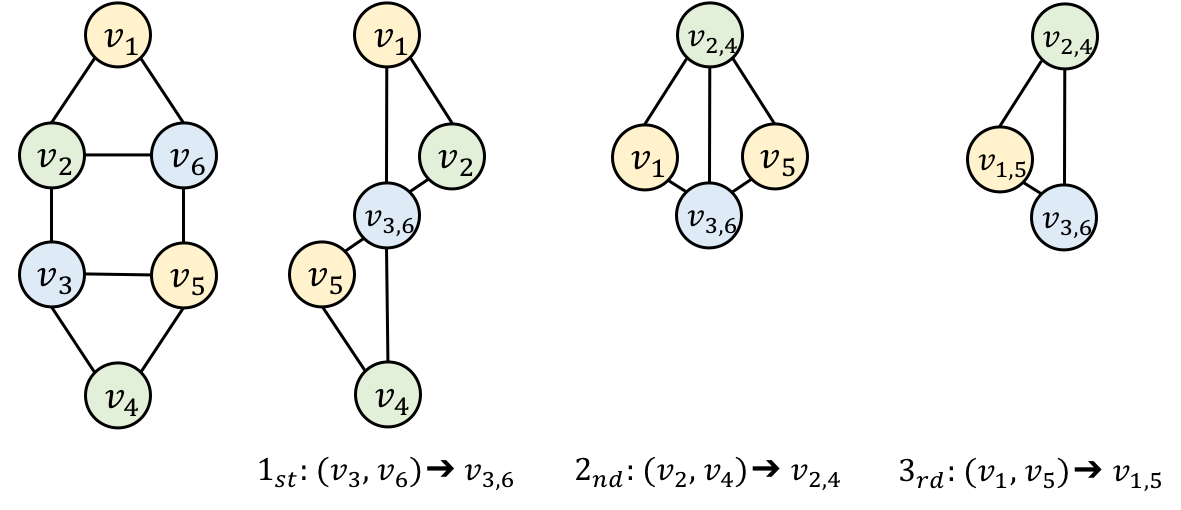}
\caption{The graph $G$ and its contractions (Case 1).}
\label{fig:d3-1-G}
\end{figure}

\begin{figure}[htb]
\centering
\includegraphics[scale=0.5]{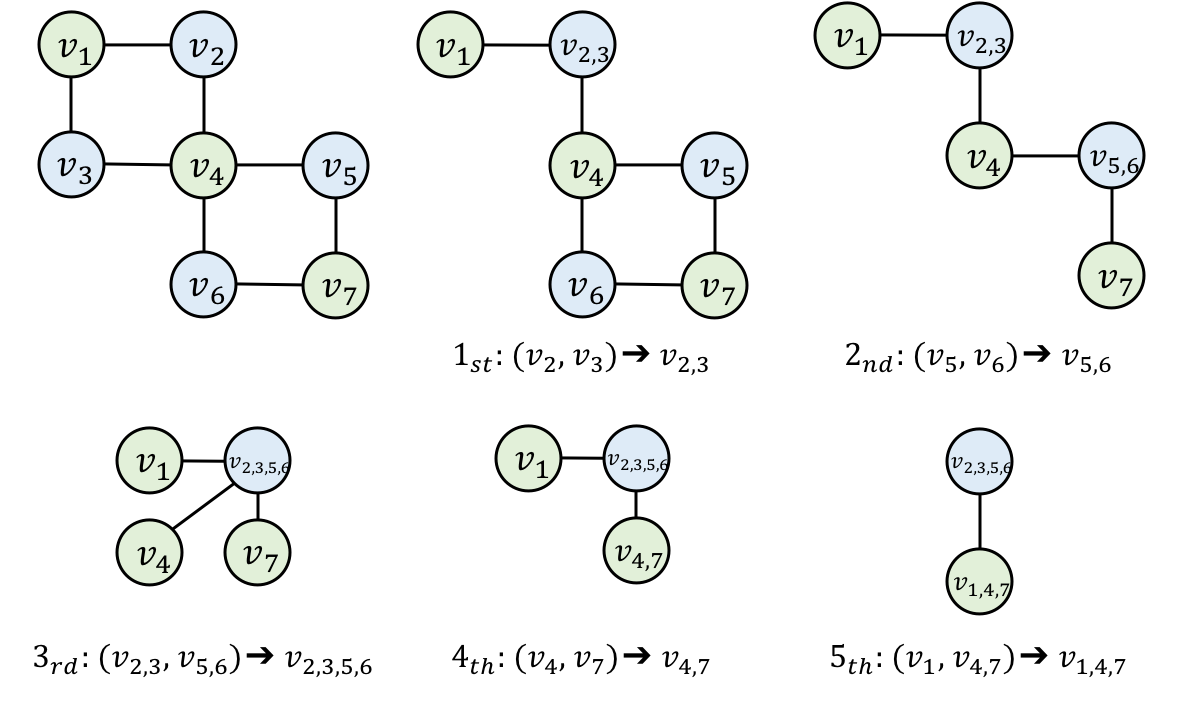}
\caption{The graph $G$ and its contractions (Case 2).}
\label{fig:d3-2-G}
\end{figure}

\subsection{Dynamic Update for Edge Insertion {in Chordal Graphs}}

\label{sec:chordal_insertion}
When it comes to chordal graphs, as they are a subset of weakly chordal graphs, Algorithms \ref{alg:weakly_insertion} and \ref{alg:weakly_deletion} apply to chordal graphs as well. Nevertheless, thanks to  the more dedicated properties of chordal graphs, the updating complexity can be further reduced. In what follows, we detail two new dynamic algorithms for chordal graphs with edge insertion and deletion.

{
Before proceeding further, we introduce some essential properties of chordal graphs.
Chordal graphs are perfect graphs. Any induced subgraph of a chordal graph is still a chordal graph.
\longdelete{
In a chordal graph $G$, a vertex $v$ of $G$ is called {\em simplicial} if and only if the subgraph induced by the vertex set $\{v\} \cup N(v)$ is a complete graph.
That is, a vertex $v$ is said simplicial if its adjacency set (i.e., close-neighborhood, \textcolor{red}{Ya-Chun, have we defined the adjacency set?? Perhaps we can move the definitions of simplicial and PEO to the beginning of this section, like two-pairs.}) is a clique.
A chordal graph $G$ on $n$ vertices is said to have a {\em perfect elimination order (PEO)} if and only if there is an ordering $\{v_1,..., v_n\}$ such that each $v_i$ is simplicial in the subgraph induced by the vertices $\{v_1,..., v_i\}$.
}
Moreover, a graph is chordal if and only if it has a perfect elimination
order~\cite{rose1970triangulated}.
Based on the above features, when considering the static procedure for minimum coloring for chordal graphs, we scan through the vertices in the reverse order of the perfect elimination order and color each vertex with the smallest color not being used among its successors. It is simply a greedy algorithm that can be colored in polynomial time if the perfect elimination order is already known~\cite{gavril1972algorithms}.
We then further consider the dynamic operation in which edge insertion is applied through adversarial perturbation with minimal updating cost considering both the perfect elimination order and the coloring (as shown in Algorithm \ref{alg:chordal_insertion}).}

\begin{algorithm}[htb]
 \caption{Dynamic Update for Edge Insertion {in Chordal Graphs}}
 \label{alg:chordal_insertion}
 \begin{algorithmic}[1]
 \renewcommand{\algorithmicrequire}{\textbf{Input:}}
 \renewcommand{\algorithmicensure}{\textbf{Output:}}
 \REQUIRE a chordal graph $G$ and an edge $(u,v) \notin E(G)$
 \ENSURE a maximum clique $K_H$ and a minimum number of colors $\left| {f_H} \right|$
 \STATE insert the edge $(u,v)$ into graph $G$;
 \STATE $H \gets G.insert(u,v)$
 \IF {$f_G(u) \neq f_G(v)$} 
 \STATE $\left| {f_H} \right| \gets \left| {f_G} \right|$; \hspace{4cm}\%{{Case 1}}
 \ELSIF{$N[u] \cap V(G)_{\geq u}$ form a clique} 
    \IF{$\left| {K_H} \right| = \left| {K_G} \right|$}
    \STATE $\left| {f_H} \right| \gets \left| {f_G}  \right|$; \hspace{3.5cm}\%{{Case 2-1}}
    \ELSE
    \STATE $\left| {f_H} \right| \gets \left| {f_G} \right|+1$; \hspace{2.8cm}\%{{Case 2-1}}
    \ENDIF
 \ELSIF{$N[u] \cap V(G)_{\geq u}$ do not form a clique}
 \STATE $\left| {f_H} \right| \gets \left| {f_G} \right|$; \hspace{4cm}\%{{Case 2-2}}
 \ENDIF
 \end{algorithmic}
\end{algorithm}

\begin{figure}[htb]
\centering
\includegraphics[scale=0.5]{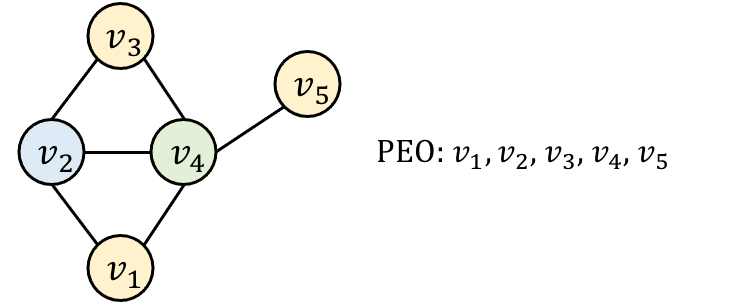}
\caption{The original chordal message conflict graph $G$.}
\label{fig:ch_0}
\end{figure}

Given the original chordal message conflict graph $G$,
an inserted edge $(u,v)$ yields a new graph $H=G+(u,v)$, which is still chordal.
Without loss of generality, assume we acquire a perfect elimination order and the coloring from the static procedure for a chordal graph $G$ (as shown in Fig.~\ref{fig:ch_0}).
After inserting an edge $(u,v)$ into the graph, we first determine whether the two endpoints of the inserted edge (i.e. vertices $u$ and $v$) have the same color in $f_G$. 
The second condition is to determine whether the neighbors of the vertex between $u$ and $v$ with the smaller index form a clique or not. 
{Notice that we always consider the vertices appeared after $u$ and $v$ in the PEO since the newly-inserted edge affects only the subsequent vertices.}

{\bf Case 1}: $f_G(u) \neq f_G(v)$, where $(u,v)=(v_2,v_5)$ (as shown in Fig.~\ref{fig:ch_1_2-1} (a)).
In the case, since $f_G(u)$ is not equal to $f_G(v)$, the maximum clique $K_H$ remains the same as $K_G$. 
The PEO is adjusted slightly while swapping the positions of $v_2$ and $v_3$, since $v_2$ is the endpoint with the smaller index of the inserted edge and $v_3$ is the vertex with the smallest index among all the neighbors greater than $v_2$ {(i.e. $N[v_2] \cap V(G)_{> v_2}$)}.
Accordingly, $\left| {f_H} \right|$ is equal to $\left| {f_G} \right|$ with the coloring unchanged.

\begin{figure}[htb]
\centering
\includegraphics[scale=0.5]{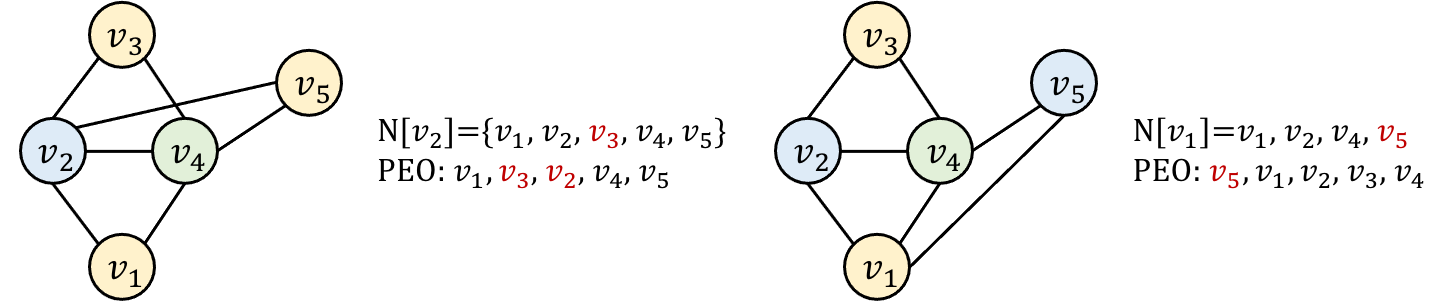}
\caption{(a) The graph $H$ and its PEO (Case 1); (b) The graph $H$ and its PEO (Case 2-1).}
\label{fig:ch_1_2-1}
\end{figure}


{\bf Case 2}: $f_G(u) = f_G(v)$.
If $f_G(u)$ and $f_G(v)$ are the same in graph $G$, either of them is forced to update to a new value due to the inserted edge.
We then consider two subcases depending on whether the neighbors of the vertex with the smaller index between $u$ and $v$ form a clique or not, concerning the graph $H$.

{

\longdelete{
\underline{Case 2-2}: $N(u)$ do not form a clique, but the neighbors of $u$ with the greater indices do form a clique, where $(u,v)=(v_3,v_5)$ (as shown in Fig.~\ref{fig:ch_3}).
In the case, we swap the positions between the endpoint of the inserted edge with the smaller index (i.e. $v_3$) and the vertex with the largest index among the neighbors greater than itself (i.e. $v_5$). 
Consequently, both the PEO and the coloring are modified marginally.
In this case, the PEO remains unchanged since the property of {\em simpliciality} is not destroyed concerning the inserted edge. 
Meanwhile, $\left| {K_H} \right|$ (resp. $\left| {f_H} \right|$) remain the same as $\left| {K_G} \right|$ (resp. $\left| {f_G} \right|$), but the color of the graph $H$ is moderately modified due to the inserted edge.

\begin{figure}[htb]
\centering
\includegraphics[scale=0.5]{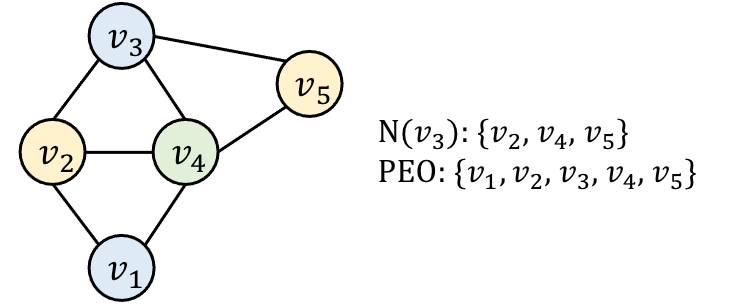}
\caption{The graph $H$ and its PEO (Case 2-2).}
\label{fig:ch_3}
\end{figure}
}

\underline{Case 2-1}: $N[u] \cap V(G)_{\geq u}$ do not form a clique. 
In this case, we have to search for a vertex among $N(u)$ that conforms to the {\em simplicial} property (note that $u$ is the endpoint with the smaller index of the inserted edge), and directly move the vertex in front of $u$ in the PEO.
\longdelete{Also noticed that the we always consider the positions after the inserted edge in the PEO since the inserted edge only affects the subsequent vertices.}

As shown in 
Fig.~\ref{fig:ch_1_2-1}~(b), for example, the inserted edge $(u,v)=(v_1,v_5)$ and $N(v_1) = \{v_2,v_4,v_5\}$.
In the original PEO, $v_1$ initially conforms to the property of simpliciality, which means ${v_1, v_2, v_4}$ form a clique. 
Nevertheless, the property of simpliciality of $v_1$ is violated due to the inserted edge $(v_1, v_5)$.
Therefore, $v_5$ is chosen and moved in front of $v_1$ in the PEO since the subgraph induced by $N[v_5]$ is a complete graph. 
The new PEO is then updated as follows: $v_5, v_1, v_2, v_3, v_4$, while $\left| {K_H} \right|$ and $\left| {f_H} \right|$ keep unchanged with the coloring only altered on $v_5$.

\underline{Case 2-2}: $N[u] \cap V(G)_{\geq u}$ form a clique.
In this case, the PEO remains unchanged since the property of {\em simpliciality} is not destroyed concerning the inserted edge. 
We then further consider whether the clique size of graph $H$ remains the same as graph $G$ or not.

Namely, if $\left| {K_H} \right|$ and $\left| {K_G} \right|$ are not equal, we have $\left| {K_H} \right| = \left| {K_G} \right|+1$ and thus $\left| {f_H} \right| = \left| {f_G} \right|+1$. 
For the scenario where $(u,v)=(v_1,v_3)$ (as shown in Fig.~\ref{fig:ch_2} (a)), the graph $H$ is still chordal since $N[v_1] \cap V(G)_{\geq v_1}$ is a complete graph.
Furthermore, the clique size of $H$ is $3+1=4$, and the coloring of $v_3$ is locally changed, where $v_3$ is the endpoint with the larger index of the inserted edge.

\begin{figure}[htb]
\centering
\includegraphics[scale=0.5]{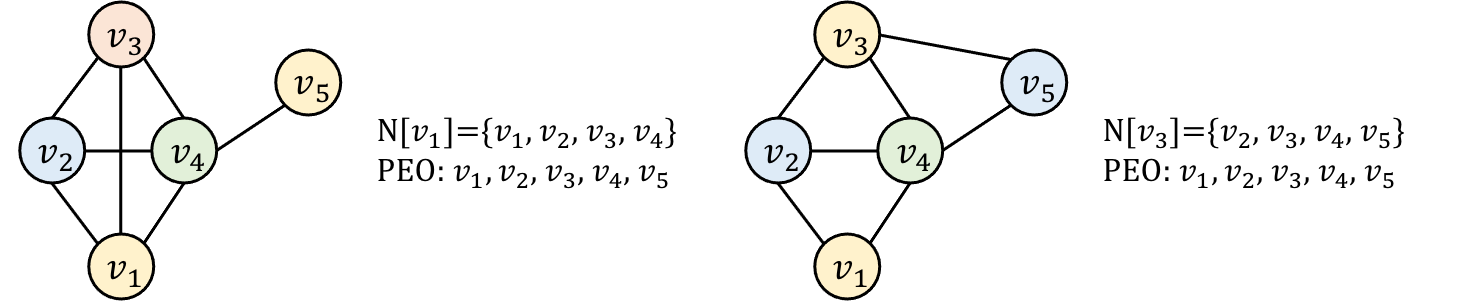}
\caption{(a) The graph $H$ where $(u,v)=(v_1,v_3)$; (b) The graph $H$ where $(u,v)=(v_3,v_5)$ (Case 2-2).}
\label{fig:ch_2}
\end{figure}

For the other scenario where $(u,v)=(v_3,v_5)$ (as shown in Fig.~\ref{fig:ch_2} (b)), $N[v_3] \cap V(G)_{\geq v_3}$ form a clique while $\left| {K_H} \right|$ remains the same as $\left| {K_G} \right|$. 
The coloring of $v_5$ is then modified locally since $v_5$ is the endpoint with the larger index of the inserted edge.

All in all, the PEO remains unchanged in this case, but the color of the graph $H$ is marginally adjusted due to the inserted edge.
}

{
\subsection{Dynamic Update for Edge Deletion {in Chordal Graphs}}
\label{sec:chordal_deletion}

\begin{algorithm}[htb]
 \caption{Dynamic Update for Edge Deletion {in Chordal Graphs}}
 \label{alg:chordal_deletion}
 \begin{algorithmic}[1]
 \renewcommand{\algorithmicrequire}{\textbf{Input:}}
 \renewcommand{\algorithmicensure}{\textbf{Output:}}
 \REQUIRE a chordal graph $H$
 \ENSURE a maximum clique $K_G$ and a minimum number of colors $\left| {f_G} \right|$
 \STATE delete the edge $(u,v)$ from graph $H$;
 \STATE $G \gets H.delete(u,v)$
 \IF {$\left| {K_G} \right| = \left| {K_H} \right|$}
 \STATE $\left| {f_G} \right| \gets \left| {f_H} \right|$;
 \hspace{4cm}\%{{Case 1}}
 \ELSE
 \STATE $\left| {f_G} \right| \gets \left| {f_H} \right|-1$;
 \hspace{3.4cm}\%{{Case 2}}
 \ENDIF
 \end{algorithmic}
\end{algorithm}

Similar to Algorithm~\ref{alg:weakly_deletion}, since $f_H(u) \neq f_H(v)$, we only need to ensure that the clique size is maintained the same or reduced by one when considering the edge deletion of chordal graphs in Algorithm~\ref{alg:chordal_deletion}.
Without loss of generality, suppose we obtain a PEO of a chordal graph $H$, denoted by $P_H$. 
We then first check the common neighbors of $u$ and $v$ in $H$, where $(u,v)$ denotes the deleted edge from $H$. If the indices of the common neighbors are smaller than the ones of $u$ and $v$ in $P_H$, obviously $P_H$ is invalid due to $(u,v)$, since it no longer satisfies the property of simpliciality. 
We then swap the positions of the common neighbor with the smallest index and the endpoint of the deleted egde with a smaller index in $P_H$ to ensure that the vertices in the newly updated PEO must be simplicial. 
The clique size of $G$ may remain the same or be reduced by one due to $(u,v)$, while the coloring may be modified marginally.


{\bf Case 1}: $\left| {K_G} \right| = \left| {K_H} \right|$.
Suppose we have Fig.~\ref{fig:ch_1_2-1} (a) as an input chordal graph $H$.
As shown in Fig. \ref{fig:ch_d} (a), the edge $(v_1,v_2)$ is deleted from $H$, and $v_4$ is the only common neighbor of $v_1$ and $v_2$ in the original $H$.
Since the index of $v_4$ in $P_H$ is greater than both $v_1$ and $v_2$, the deletion of $(v_1,v_2)$ does not affect the simplicial characteristic of $v_4$. 
Thus, PEO, $\left| {K_G} \right|$, $\left| {f_G} \right|$ and the coloring remain unchanged.

\begin{figure}[htb]
\centering
\includegraphics[scale=0.5]{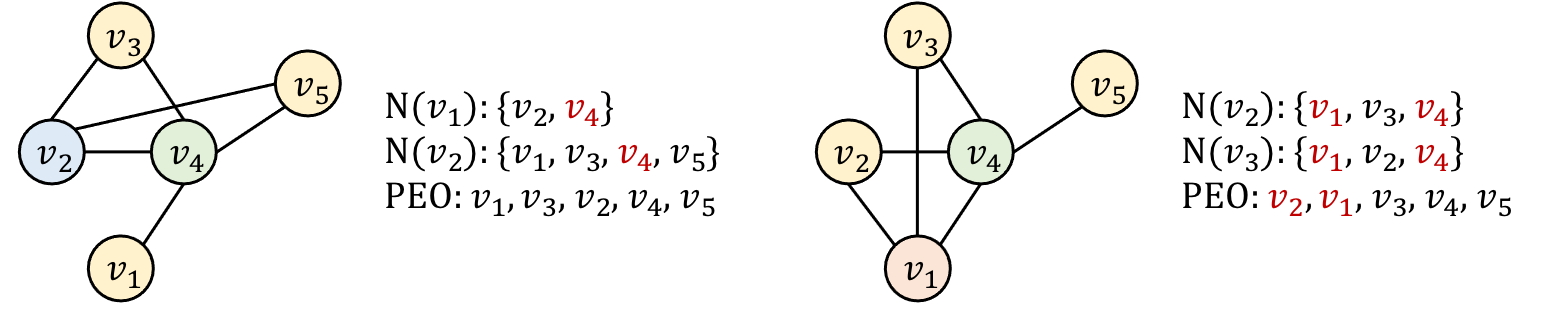}
\caption{(a) The graph $G$ and its PEO (Case 1); (b) The graph $G$ and its PEO (Case 2).}
\label{fig:ch_d}
\end{figure}

{\bf Case 2}: $\left| {K_G} \right| \neq \left| {K_H} \right|$.
We illustrate the other case that $\left| {K_G} \right| \neq \left| {K_H} \right|$ by using Fig.~\ref{fig:ch_2} (a) as an input chordal graph $H$.
As shown in Fig. \ref{fig:ch_d} (b), the edge $(v_2,v_3)$ is deleted from $H$, while $v_1$ and $v_4$ are the common neighbors of $v_2$ and $v_3$.
Nevertheless, the index of $v_1$ in $P_H$ is smaller than both $v_2$ and $v_3$, while the index of $v_4$ is greater than $v_2$ and $v_3$ on the contrary.
In other words, we ignore $v_4$ but swap the positions of $v_1$ and $v_2$ in $P_H$, since $v_1$ is the common neighbor whose index is smaller than $v_2$ and $v_3$, while $v_2$ is the vertex with the smaller index of $(v_2,v_3)$.
In the meantime, $\left| {K_G} \right|=\left| {f_G} \right|=3$, which is reduced by one.
Also noticed that the coloring is altered merely on $u$ and it will be changed to the same color as $v$, where $u$ denotes the vertex with the smaller index of the deleted edge.

}

\section{Proofs}
\label{sec:proof}
\subsection{Proof of Theorem~\ref{thm:optimality}}

We first show that Algorithms~\ref{alg:weakly_insertion} and \ref{alg:chordal_insertion} can correctly output the optimal vertex coloring when an edge is inserted, then show that such vertex coloring achieves the information-theoretic optimality in chordal networks with chordal or weakly chordal conflict graphs.

First, Algorithm~\ref{alg:weakly_insertion} builds upon the sequence of two-pairs identified by Algorithm~\ref{alg:weakly_static}, through removing the violated two-pairs due to edge insertion, and adding new valid two-pairs in the perturbed graph. According to the procedure of Algorithm~\ref{alg:weakly_insertion}, the updated sequence of two-pairs for $H$ corresponds to one of possible outputs of Algorithm~\ref{alg:weakly_static} for the perturbed graph $H$. 

The \emph{Correctness} Theorem~\cite{hayward1989optimizing}
shows that 
Algorithm~\ref{alg:weakly_static} works for an arbitrary contraction order of two-pairs in a weakly chordal graph~$G$. 
It eventually finds a largest clique as well as the corresponding minimum coloring of $G$ 
based on the perfect graph property. 
Due to the assumption that the weak chordality of graph $H$ is still maintained after inserting an edge into graph $G$, 
Algorithm \ref{alg:weakly_insertion} is correct, and still achieves the optimal vertex coloring in the weakly chordal graph $H$. 

{
Furthermore, for Algorithm~\ref{alg:chordal_insertion}, it is shown in \cite{gavril1972algorithms} that the minimum vertex coloring problem can be optimally solved in polynomial time when the input graph is a chordal graph. As chordal graphs are perfect, the maximum clique size is equal to the chromatic number of a chordal graph. 
An optimal coloring is then obtained by applying a greedy coloring algorithm to the vertices in the reverse order of an arbitrary perfect elimination order (PEO).
In the process of edge insertion or deletion, PEO is guaranteed to be the basis of minimum vertex coloring through modifying 
the order of vertices in PEO while satisfying the simplicial properties.
Therefore, Algorithm~\ref{alg:chordal_insertion} is correct, and still achieves the optimal vertex coloring in the chordal graph $H$.
}

Second, according to \cite{yi2018tdma}, optimal vertex coloring on the weakly chordal message conflict graph (that includes chordal graphs as special cases) leads to  the information-theoretic optimality. Algorithm~\ref{alg:weakly_static} yields a sequence of two-pairs, for which the sequential color assignment results in the optimal vertex coloring for weakly chordal graphs. As such, the updated sequence of two-pairs through Algorithms~\ref{alg:weakly_insertion} and~\ref{alg:weakly_static} offers the optimal color assignment, and therefore achieves the information-theoretical optimality. 
It applies to Algorithm~\ref{alg:chordal_insertion}, as the updated PEO yields the optimal color assignment, which is also information-theoretically optimal.


\begin{remark} [Edge deletion]
The proof of Theorem~\ref{thm:optimality} works for edge deletion as well, because Algorithms \ref{alg:weakly_deletion} and \ref{alg:chordal_deletion} can correctly output the optimal vertex coloring when an edge is deleted while chordality maintains, and such optimal coloring yields information-theoretical optimality for chordal and weakly chordal graphs.
\end{remark}

\subsection{Proof of Theorem \ref{thm:complexity}}
One can observe that the update cost caused by dynamic operations may vary. In order to bound the number of two-pairs updated by our dynamic algorithm, we explore the intersection representation properties of a weakly chordal graph. Note that a message conflict graph is the square of the line graph of a network topology graph in the TIM problem.

Particularly, the intersection representation of a graph is important in many real world applications. 
For example, a tree decomposition represents the vertices of a given graph as subtrees of a tree, 
in which vertices in the given graph are adjacent if and only if the corresponding subtrees intersect in the tree. 
More formally, 
the intersection representation of a graph on a tree is generally defined as follows. 
An $(h, s, t)$-representation of a given graph $G$ comprises a collection of subtrees $\{S_v |v \in V(G)\}$ of a tree $T$ , such that
\subsubsection{} 
The maximum degree of tree $T$ is at most $h$.
\subsubsection{}
The maximum degree of each subtree $S_v$ is at most $s$.
\subsubsection{}
The two corresponding subtrees $S_u$ and $S_v$ in $T$ have at least $t$ vertices in common if and only if there is an edge between two vertices $u$ and $v$ in the graph $G$.

First, we show that every message conflict graph has a $(4,4,2)$-representation; i.e., it is a $(4,4,2)$ graph. 
Recall the forbidden structures of a $(4, 4, 2)$ graph, reported in~\cite{golumbic2009intersection}. 

\begin{theorem} \cite{golumbic2006finding}
Let $G$ be a weakly chordal $(K_{2,3}$, $\overline{4P_2}$, $\overline{P_2 \cup P_4}$, $\overline{P_6}$, $\overline{H_1}$, $\overline{H_2}$, $\overline{H_3})$-free graph, then the \textbf{Construct (4, 4, 2)-representation algorithm} finds a $(4, 4, 2)$-representation of $G$.
\end{theorem}

If $G$ is a $(4, 4, 2)$ graph, then $G$ is weakly chordal and contain no such induced subgraphs like $K_{2,3}$, $\overline{4P_2}$, $\overline{P_2 \cup P_4}$, $\overline{P_6}$, $\overline{H_1}$, $\overline{H_2}$, and $\overline{H_3}$, 
{
which are shown in the left-hand side of subfigures in Fig. \ref{fig:mix}.
}

As mentioned, a message conflict graph $G$ is the square of the line graph of a graph, which implies that 
$G$ cannot contain the forbidden structures of a line graph. We revisit the recognition of a line graph and 
recall the following theorem. 

\begin{theorem} \cite{vsoltes1994forbidden}
$G$ is a line graph if and only if $G$ does not contain any of the graphs $G_1$-$G_9$ as an induced subgraph.
\end{theorem}

In the right-hand side of these figures, they are induced subgraphs of forbidden structures of $(4, 4, 2)$ graph. 
Note that they are also the minimal forbidden induced subgraphs for line graph \cite{vsoltes1994forbidden}.
Accordingly, the forbidden structures of $(4, 4, 2)$ graph mentioned above are not contained in the scenario we mainly discuss about, since TDMA on the network topology is equivalent to graph coloring on its square of line graph.

\begin{figure}
	\centering
	\subfigure[$K_{2,3}$ and $G_6$]{
		\begin{minipage}{7.5cm} 
            \includegraphics[width=\textwidth]{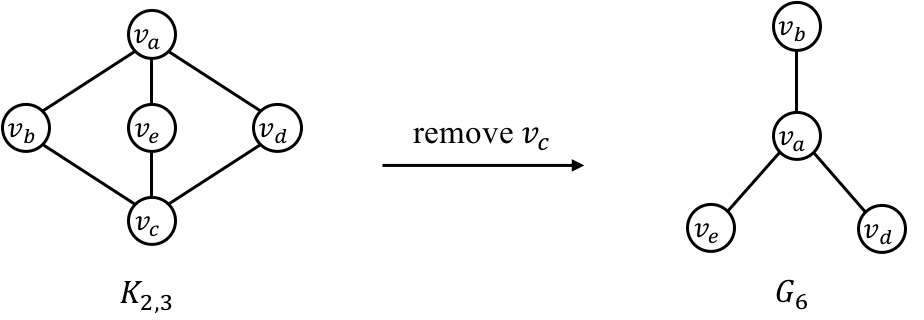} \\
		\end{minipage}
	}
	\subfigure[$\overline{4P_2}$ and $G_7$]{
		\begin{minipage}{7.5cm}
			\includegraphics[width=\textwidth]{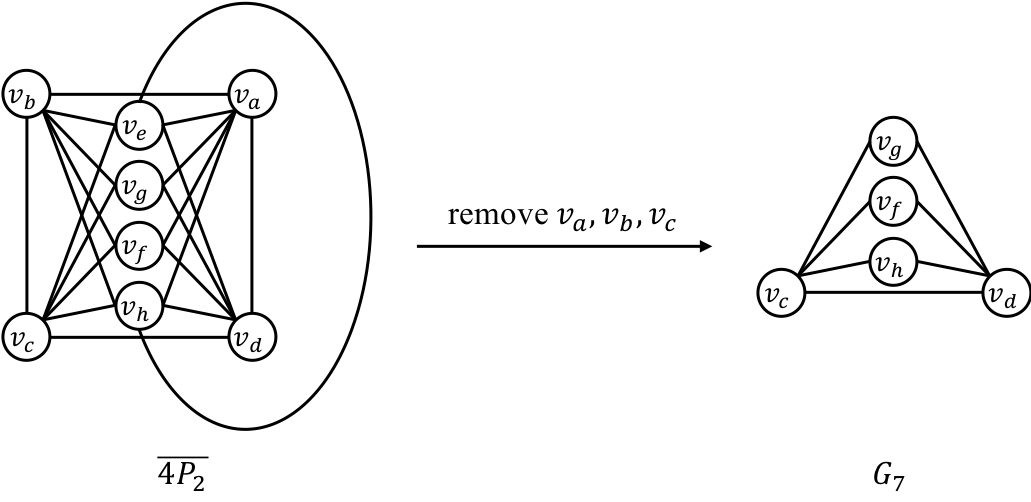} \\
		\end{minipage}
	}
	\subfigure[$\overline{P_2 \cup P_4}$ and $G_3$]{
		\begin{minipage}{7.5cm} 
            \includegraphics[width=\textwidth]{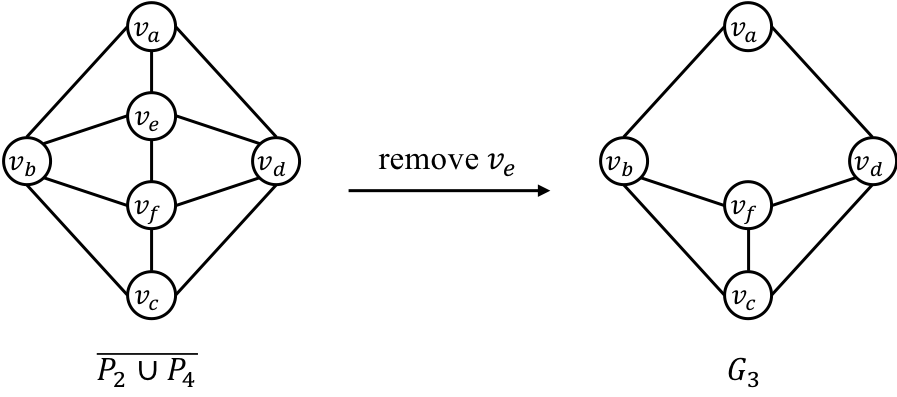} \\
		\end{minipage}
	}
	\subfigure[$\overline{P_6}$ and $G_3$]{
		\begin{minipage}{7.5cm} 
            \includegraphics[width=\textwidth]{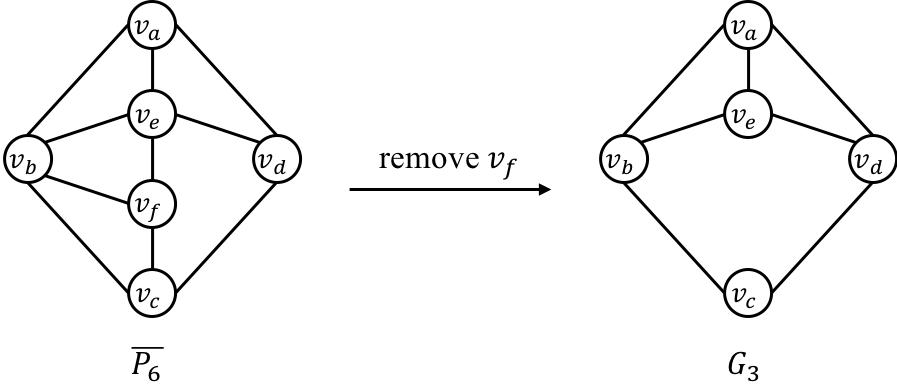} \\
		\end{minipage}
	}
	\subfigure[$\overline{H_1}$ and $G_5$]{
		\begin{minipage}{7.5cm} 
            \includegraphics[width=\textwidth]{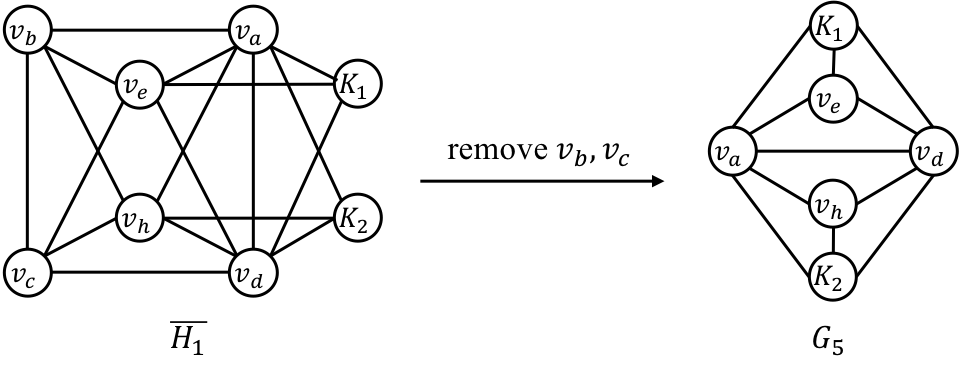} \\
		\end{minipage}
	}
	\subfigure[$\overline{H_2}$ and $G_5$]{
		\begin{minipage}{7.5cm} 
            \includegraphics[width=\textwidth]{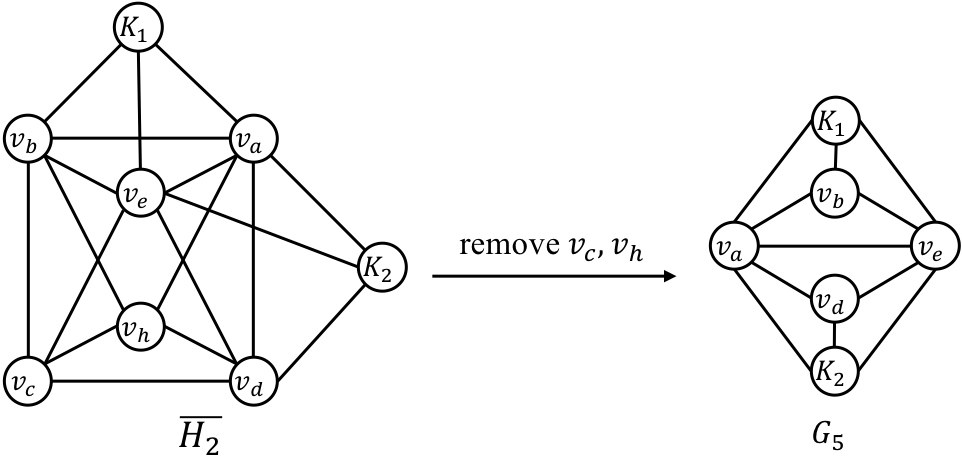} \\
		\end{minipage}
	}
	\subfigure[$\overline{H_3}$ and $G_4$]{
		\begin{minipage}{7.5cm} 
            \includegraphics[width=\textwidth]{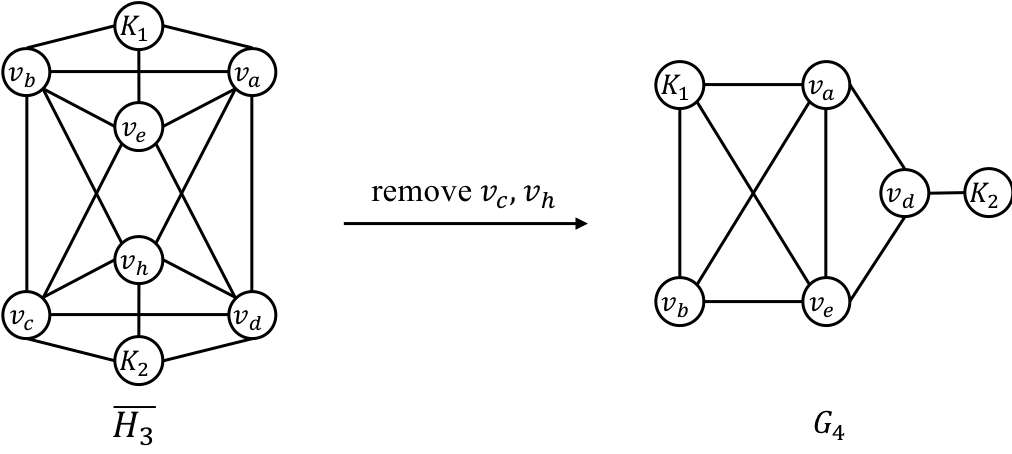} \\
		\end{minipage}
	}
		\caption{Forbidden structures of $(4, 4, 2)$ graph and the induced subgraphs} 
		\label{fig:mix}
\end{figure}


\longdelete{
\begin{figure}[htb]
\centering
\includegraphics[scale=0.6]{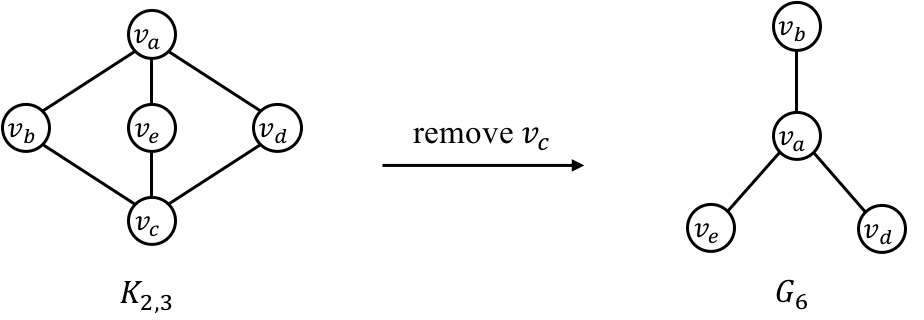}
\caption{$K_{2,3}$ and $G_6$}
\label{fig:forbidden graph 1}
\end{figure}

\begin{figure}[htb]
\centering
\includegraphics[scale=0.6]{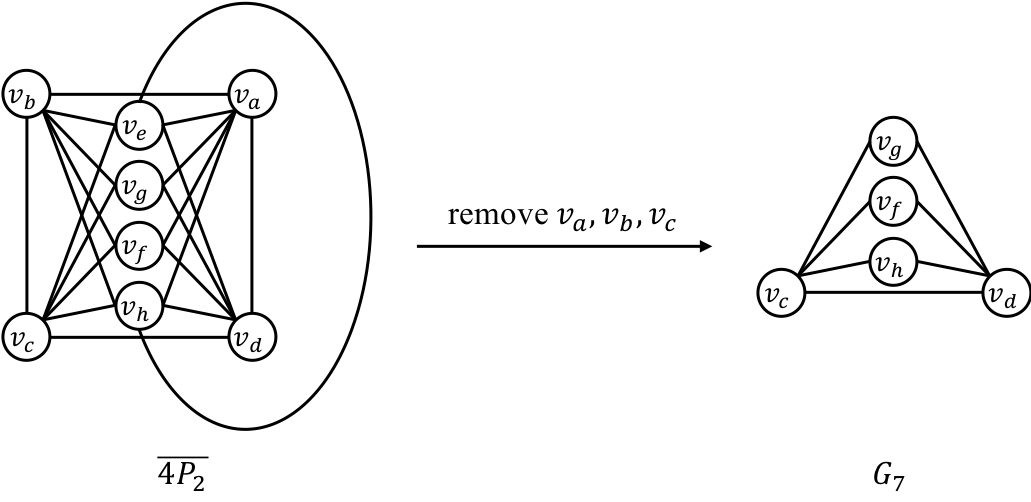}
\caption{$\overline{4P_2}$ and $G_7$}
\label{fig:forbidden graph 2}
\end{figure}

\begin{figure}[htb]
\centering
\includegraphics[scale=0.6]{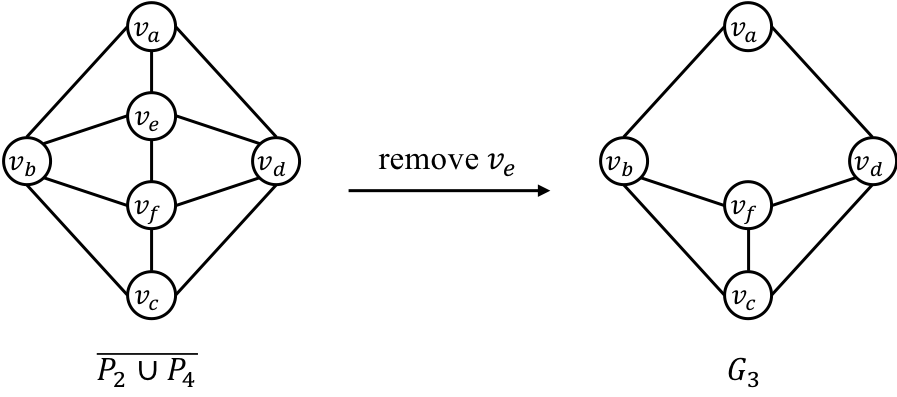}
\caption{$\overline{P_2 \cup P_4}$ and $G_3$}
\label{fig:forbidden graph 3}
\end{figure}

\begin{figure}[htb]
\centering
\includegraphics[scale=0.6]{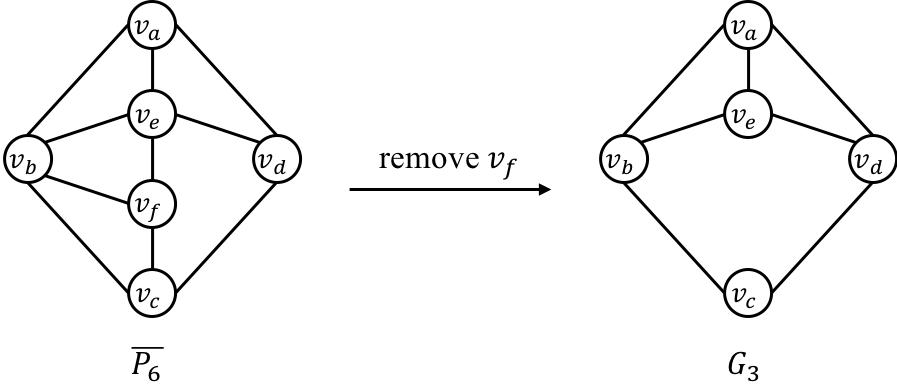}
\caption{$\overline{P_6}$ and $G_3$}
\label{fig:forbidden graph 4}
\end{figure}

\begin{figure}[htb]
\centering
\includegraphics[scale=0.6]{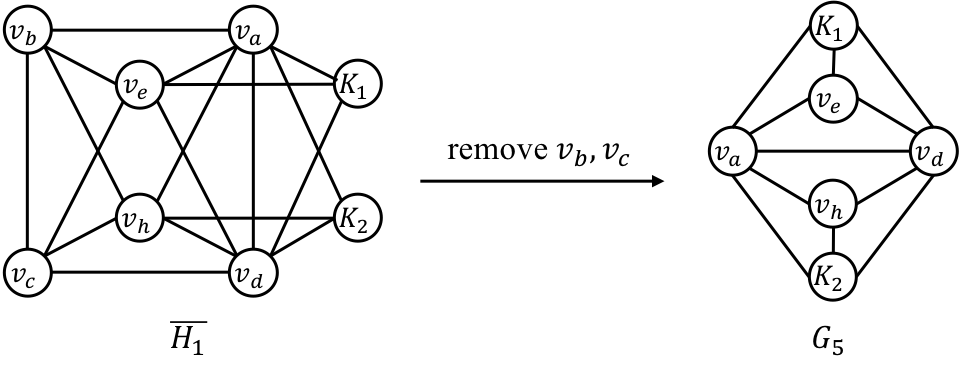}
\caption{$\overline{H_1}$ and $G_5$}
\label{fig:forbidden graph 5}
\end{figure}

\begin{figure}[htb]
\centering
\includegraphics[scale=0.6]{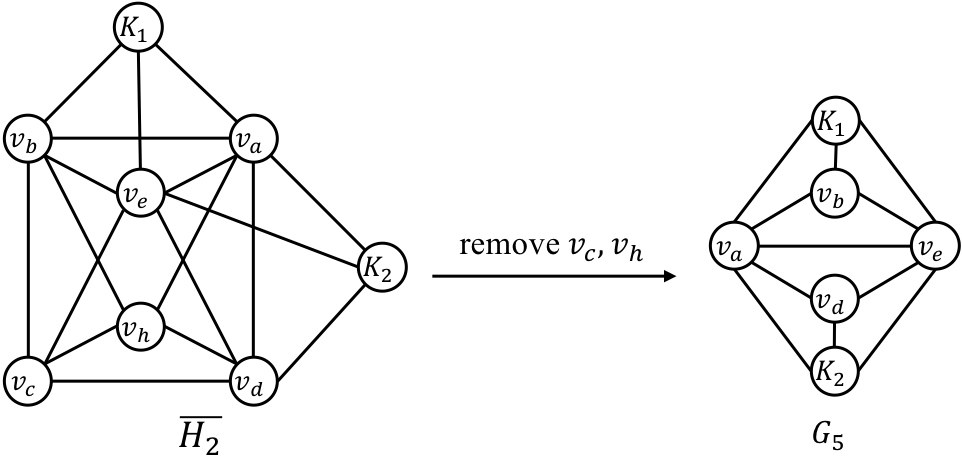}
\caption{$\overline{H_2}$ and $G_5$}
\label{fig:forbidden graph 6}
\end{figure}

\begin{figure}[htb]
\centering
\includegraphics[scale=0.5]{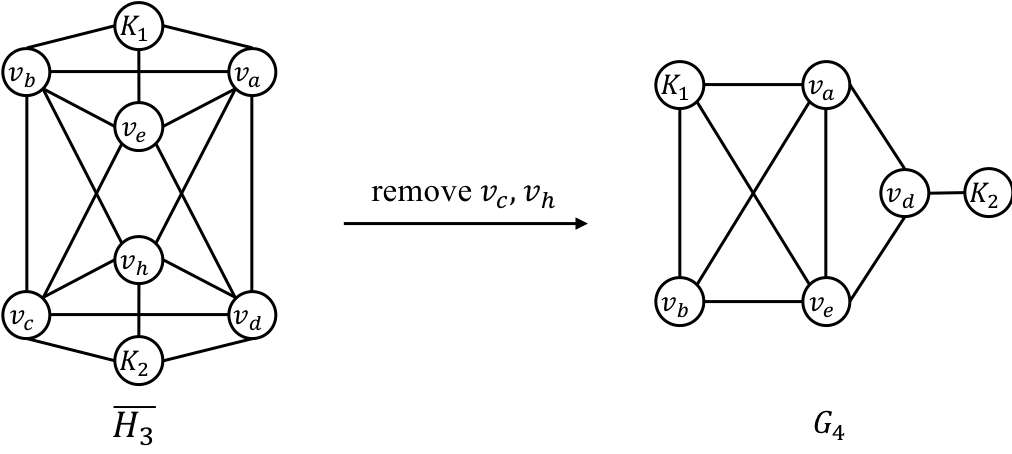}
\caption{$\overline{H_3}$ and $G_4$}
\label{fig:forbidden graph 7}
\end{figure}
}

\longdelete{
\begin{figure}[h]
\centering
\includegraphics[scale=0.5]{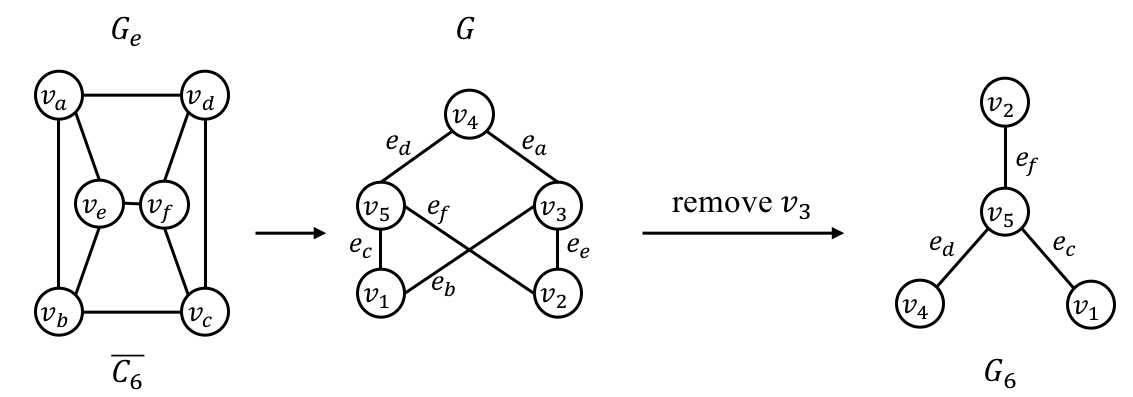}
\caption{$\overline{C_6}$ and $G_6$}
\label{fig:forbidden graph 8}
\end{figure}
}

Based on the above observation, a message conflict graph $G$ is a $(4,4,2)$ graph. Therefore, the procedure of Construct-$(4,4,2)$-representation algorithm can recognize $G$ and construct its intersection representation. 

\begin{lemma}
\label{lemma:1}
Given a $(4,4,2)$ graph, the number of two-pairs needed to be updated, due to an edge insertion, is bounded by a constant. 
\end{lemma}

\begin{proof}
The constant number of updates is owing to the properties of a two-pair in a $(4,4,2)$-representation. 
Fig.~\ref{fig:Intersection tree graph 1} shows an example of a $(4,4,2)$-representation of a two-pair, i.e. $v_a$ and $v_d$. 
Here the maximum degree of the tree in the right-hand side is four. 
When inserting an edge $(u,v)$ in Algorithm~\ref{alg:weakly_insertion}, one can observe that the number of updates of two-pairs is bounded by the total number of neighbors of $u$ and $v$. More precisely, the update number is bounded by the number of the neighbors that can possibly form a two-pair. As shown in Fig.~\ref{fig:Intersection tree graph 3-2}, due to the maximum degree of 4, $u$ has at most four neighbors each of which can form a two-pair, and so does $v$. Therefore, this completes the proof of the lemma as well as the second main result regarding the number of re-coloring updates.  
\end{proof}

\begin{remark} [Edge deletion]
The proof of Theorem~\ref{thm:complexity} can be amended accordingly for edge deletion with  Algorithms \ref{alg:weakly_deletion} and \ref{alg:chordal_deletion}.
{When an edge is deleted, the intersection representation properties of a weakly chordal graph hold as well if chordality  maintains. 
Thus the number of two-pairs to be updated is still bounded by a constant. Moreover, this result can be applied to chordal graphs since chordal graphs are a subclass of weakly chordal graphs.}
\end{remark}

\begin{figure}[htb]
\vspace{-20pt}
\centering
\includegraphics[scale=0.15]{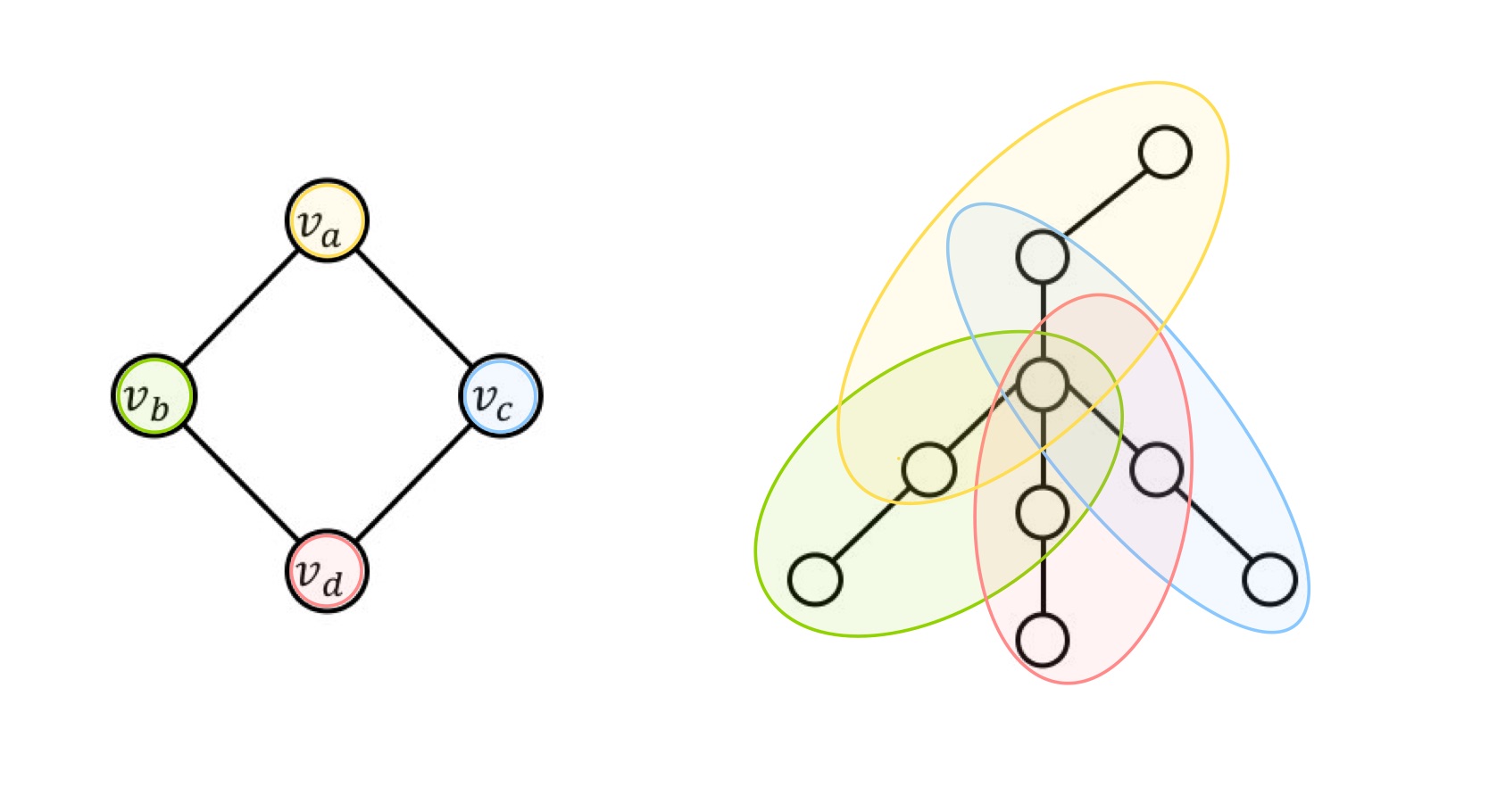}
\vspace{-20pt}
\caption{A two-pair of $v_a$ and $v_d$ and its intersection (tree) representation}
\label{fig:Intersection tree graph 1}
\end{figure}
%
%
\begin{figure}[htb]
\centering
\vspace{-20pt}
\includegraphics[scale=0.15]{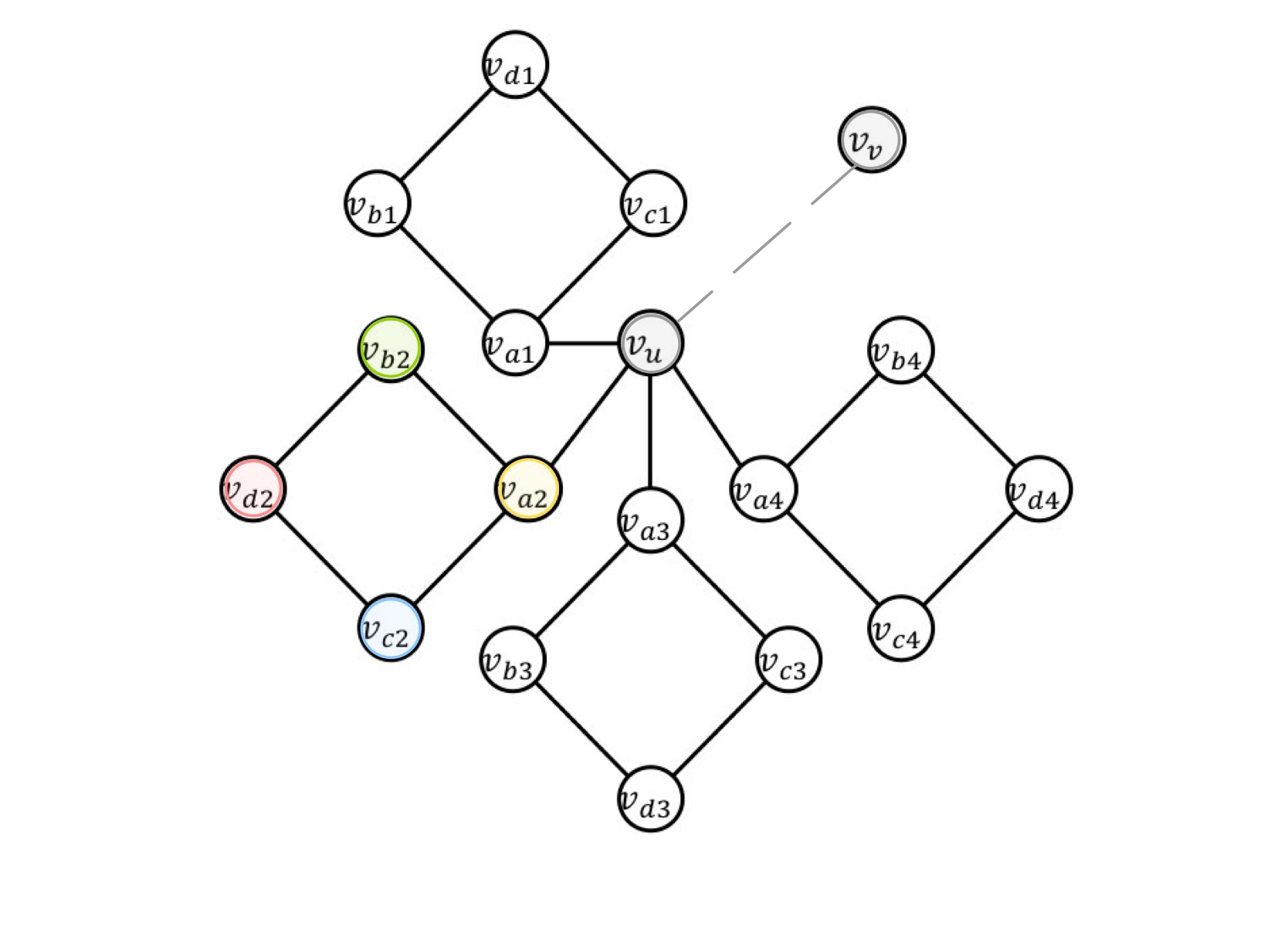}
\vspace{-20pt}
\caption{At most four possible adjacent two-pairs of $v_u$}
\label{fig:Intersection tree graph 3-1}
\end{figure}

\begin{figure}[ht]
\centering
\vspace{-20pt}
\includegraphics[scale=0.15]{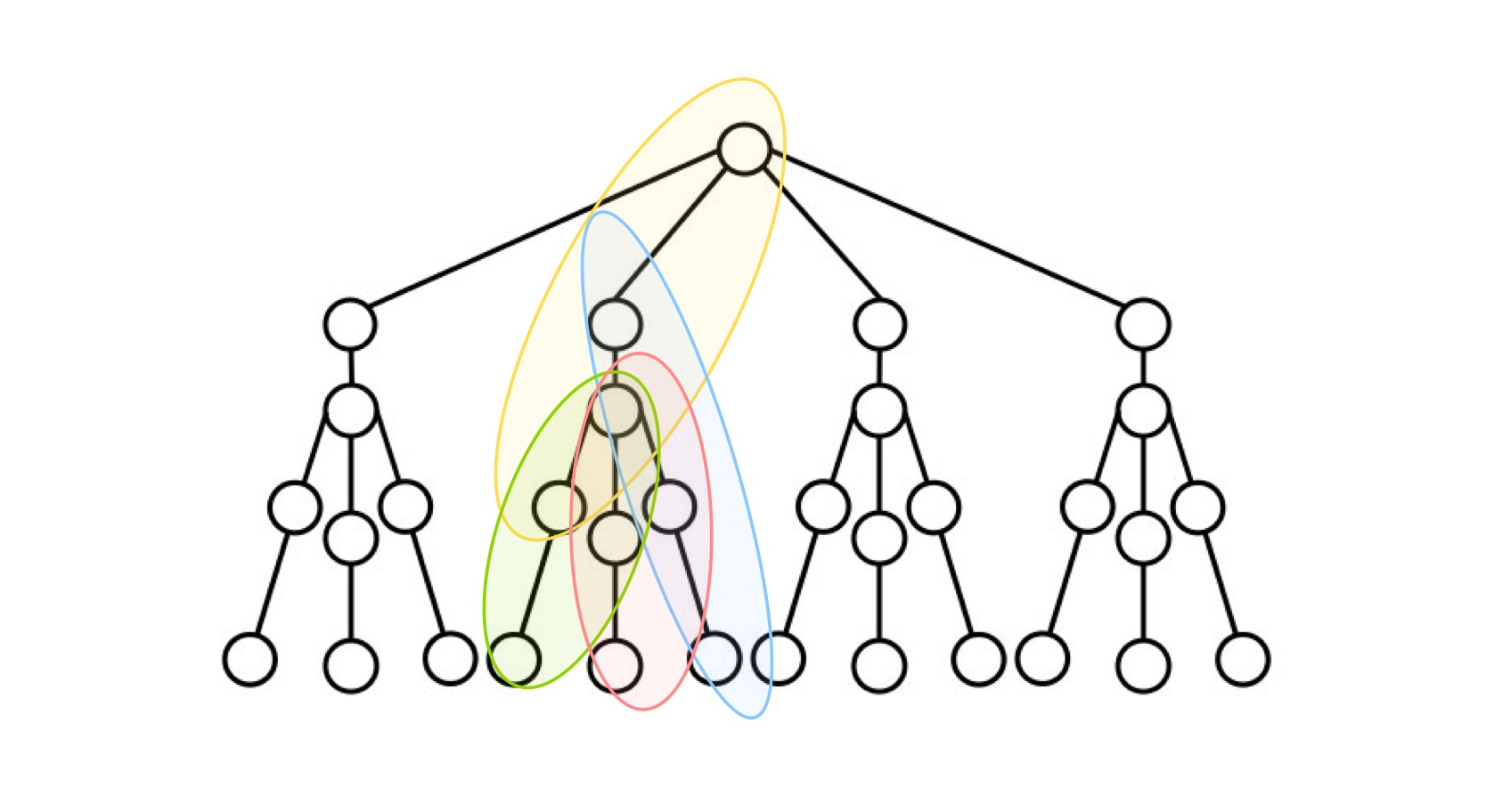}
\vspace{-20pt}
\caption{An illustration of the intersection (tree) representation of Fig.~\ref{fig:Intersection tree graph 3-1}}
\label{fig:Intersection tree graph 3-2}
\end{figure}

\section{Conclusion}
We have studied the topological interference management problem in a dynamic setting from an algorithmic perspective, where an adversary perturbs the chordal network topology that results in edge insertion/deletion on message conflict graphs. By exploiting the structural properties of {message conflict graphs in chordal graph classes}, we have proposed a set of dynamic graph re-coloring algorithms to robustify the optimal coding scheme against such adversarial perturbation. The proposed algorithm makes a constant number of re-coloring updates for each single edge insertion/deletion on conflict graphs, regardless of the sizes of the network or message set, whilst still yielding
the information-theoretic optimality under the TIM setting.



\renewcommand\refname{Reference}
\bibliographystyle{IEEEtran}
\bibliography{reference}


\end{document}